\documentclass[twocolumn,pra,nofootinbib,superscriptaddress,longbibliography]{revtex4-2}

\usepackage[dvips]{graphicx} 
\usepackage{amsfonts}
\usepackage{amssymb}
\usepackage{amscd}
\usepackage{amsmath}    
\usepackage{mathrsfs}
\usepackage{enumerate}
\usepackage{epsfig}
\usepackage{subfigure}
\usepackage{subfloat}
\usepackage{xcolor}
\usepackage{amsthm}
\usepackage{physics}
\usepackage{multirow}
\usepackage[most]{tcolorbox}
\usepackage{tikz}
\usetikzlibrary{quantikz}
\usetikzlibrary{backgrounds}
\usepackage{tabularx}
\usepackage{float}
\usepackage{appendix}
\usepackage{makecell}
\usepackage{algorithm}
\usepackage{algorithmicx}
\usepackage{algpseudocode}
\usepackage{dsfont}
\usepackage[colorlinks, linkcolor=red, anchorcolor=blue, citecolor=green]{hyperref}
\usepackage{MnSymbol}
\newtcolorbox[auto counter]{mybox}[2]{
enhanced,
breakable,
label=#1,
colback=blue!5!white,
colframe=blue!75!black,
fonttitle=\bfseries,
title=Box \thetcbcounter: #2
}

\newtheorem{theorem}{Theorem}
\newtheorem{lemma}{Lemma}
\newtheorem{corollary}{Corollary}

\newtheorem{definition}{Definition}

\newcommand{\id}{\mathbb{I}}

\newcommand{\comments}[1]{}
\definecolor{mintcyan}{RGB}{180, 240, 220}

\begin{document}
\title{Can scrambling protect quantum state distinguishability under noise?}
\author{Guoding Liu}
\affiliation{Center for Quantum Information, Institute for Interdisciplinary Information Sciences, Tsinghua University, Beijing, 100084 China}
\author{Chushi Qin}
\affiliation{Center for Quantum Information, Institute for Interdisciplinary Information Sciences, Tsinghua University, Beijing, 100084 China}
\author{Zitai Xu}
\affiliation{Joint Center for Quantum Information and Computer Science, University of Maryland, College Park, Maryland 20742, USA}
\author{Xiongfeng Ma}
\email{xma@tsinghua.edu.cn}
\affiliation{Center for Quantum Information, Institute for Interdisciplinary Information Sciences, Tsinghua University, Beijing, 100084 China}
\author{Zi-Wen Liu}
\email{zwliu0@tsinghua.edu.cn}
\affiliation{Yau Mathematical Sciences Center, Tsinghua University, Beijing 100084, China}

\begin{abstract}

Quantum state distinguishability is a fundamental concept in quantum information science that underpins a wide range of important practical tasks. Traditionally formulated for pairs of states, quantum state distinguishability is here extended to quantum state ensembles, which we characterize through the average pairwise trace distance. Motivated by both theoretical and practical interest in noisy quantum information processing, we ask whether ``minimally'' scrambled ensembles modeled by 2-designs protect distinguishability under noise, which sheds light on the fundamental competition between noise and information scrambling.  Using a rigorous decoupling approach, we establish tight bounds on noisy ensemble distinguishability. We show that the distinguishability of noisy 2-design ensembles exhibits a sharp threshold and phase-transition behavior governed by channel conditional entropy: below the threshold, the states remain mutually distinguishable with high probability, while above it, distinguishability undergoes a sudden power-law decay and then collapses exponentially.  On the other hand, under local purity-shrinking noise, post-measured noisy 2-design ensembles become exponentially indistinguishable for any measurement, precluding a noise threshold for learning tasks such as shadow tomography.
These results reveal a sharp difference between unmeasured and post-measured scrambled ensembles: the former can retain high distinguishability for sufficiently small noise, whereas the latter exhibits no such protected regime. We discuss the implications of these results for crucial tasks ranging from quantum communication and cryptography to learning.

\end{abstract}

\maketitle

\section{Introduction}

Quantum state distinguishability is a cornerstone of quantum information science, and its operational characterization by the trace distance is a foundational result of the field~\cite{HOLEVO1973Statistical,Helstrom1969estimation,nielsen2002quantum}.
High distinguishability is essential for reliably encoding, transmitting, and extracting information from quantum systems, thereby constituting a fundamental resource behind quantum advantages in information processing.

In particular, the global distinguishability properties of state ensembles, beyond the traditional formulation for pairs of states, naturally underlie many practical tasks across quantum communication, characterization, cryptography, tomography, learning, among others.
In communication and cryptography, mutually distinguishable quantum codebooks are essential for reliably encoding distinct classical messages, enabling quantum advantages in tasks such as quantum fingerprinting~\cite{Buhrman2001fingerprinting}. 
In quantum computing and learning, such codebooks can serve as hypothesis sets for state identification, where mutual distinguishability among the elements is necessary for achieving efficient sample complexity~\cite{Aaronson2007learnability}. 
High distinguishability is also crucial when an ensemble is used as a measurement basis for extracting properties of quantum states, as in classical shadow tomography~\cite{huang2020shadow}. 
Beyond pursuing high distinguishability itself, exploiting the gap between global and restricted measurements leads to the important application of quantum data hiding~\cite{DiVincenzo2002hiding,Hayden2004Randomizing,Matthews2009Distinguishability}.


As a central challenge in quantum technologies, noise may naturally degrade the distinguishability of quantum state ensembles.
For example, macroscopic superpositions such as GHZ-like states are fragile: arbitrarily weak local noise can destroy their distinguishability~\cite{Dur2004Macroscopic,arunachalam2023optimalalgorithmslearningquantum}. 
Can one design state ensembles whose distinguishability is robust against local noise? 
The widely studied phenomenon of information scrambling provides a compelling mechanism for such robustness, as it delocalizes logical information and thereby reduces its susceptibility to local errors. 
This intuition is closely related to how scrambling generates global entanglement~\cite{Page1993average,Nahum2017entanglement,Hosur2016Chaos,Liu_2018_beyond_scrambling} and approximate quantum error correction~\cite{gottesman1997stabilizer,Brown2013short,Gullans2021LowDepth,Kong2022CQEC,Darmawan2024Lowdepth,Yi_2024,Nelson2025FTlowdepth,Liu2026AQEC}. 
At the same time, the mechanism that provides protection also spreads local errors into nonlocal ones and thus exacerbates them, creating a nontrivial competition between information scrambling and noise amplification. 
This competition has deep physical significance~\cite{Choi2020QECMIPT,Morvan2024phase} and is also of intensive interest in the study of quantum computational supremacy~\cite{Aharonov_2023_noisyRCS,Fefferman_2024_nonunitalRCS}.
Rigorously understanding this competition not only offers new insight into the fundamental understanding of the noise-scrambling interplay, but also provides the foundation for practically harnessing scrambling to preserve distinguishability under noise.

Furthermore, from a practical perspective, one would clearly desire to realize the scrambling protection with minimal complexity overhead.
It is well known that the second moment of the Haar measure marks the  onset of essential information-theoretic features of scrambling, notably global entanglement generation, decoupling, and chaotic behaviors~\cite{Hosur2016Chaos,Hayden2007Mirrors,Liu_2018_beyond_scrambling,Liu_2018_quantum_designs,Szehr2013Decoupling,Roberts2017Design}, rendering \(2\)-design ensembles the standard ``minimal'' model of scrambling~\cite{Emerson_2003_pseudorandom,Dankert2009design}. 
Crucially, \(2\)-designs are efficiently realizable in various settings with low circuit depth or evolution time~\cite{Harrow2009RandomDesign,Brandao2016ApproximateDesign,Cleve2016NearLinearDesign,Dankert2009design,kueng2015qubitstabilizerstatescomplex,zhu2017threedesign,Webb2016Clifford3design,Bravyi2021Clifford,chen2025Incompressibility,Li2024SymmetricPseudorandomness,schuster2024randomunitariesextremelylow,laracuente2024approximateunitarykdesignsshallow}, which serve as accessible sources of scrambling and randomness that enable many quantum information tasks, including classical shadow tomography~\cite{huang2020shadow}, randomized benchmarking~\cite{Knill2008RB,Elben2023toolbox}, quantum metrology~\cite{Zhou2026Metrology}, and random circuit sampling~\cite{arute2019supremacy,Bouland2019sampling,YulinWu2021Superconducting}. 
Against this backdrop, understanding how noise degrades distinguishability in \(2\)-design ensembles becomes a central problem, both for clarifying the physical boundary between information scrambling and quantum noise, and for determining the robustness and viability of design-based protocols in real-world applications.

In this work, we systematically study this problem and establish a comprehensive characterization of how noise affects the distinguishability of 2-design ensembles, drawing on rigorous information-theoretic methods. 
As a framework that may be of independent use, we first extend the traditional notion of distinguishability from pairs to an ensemble of arbitrary size using the average pairwise trace distance. 
Technically, we need to analyze how the average pairwise distance behaves under a noise channel \(\mathcal{N}\) for the noisy ensemble \(\{p_x,\mathcal{N}(\rho_x)\}\). 
We achieve this by establishing connections with decoupling methods from quantum Shannon theory: the average pairwise trace distance can be mapped to an decoupling error of the ensemble, yielding tight upper and lower bounds. 
The strong decoupling properties of \(2\)-designs~\cite{Dupuis2014Decoupling} allow these bounds to be expressed compactly in terms of conditional entropies of the noise channel \(\mathcal{N}\)~\cite{rubboli2026quantumconditionalentropiesconvex}.

Evaluating the entropy for concrete noise models enables us to identify rigorous noise thresholds for the distinguishability of $n$-qubit 2-design ensembles. Specifically, we establish two positive constant thresholds governed by the single-letter conditional von Neumann entropy $H(\mathcal{N})$ and the conditional collision entropy $H_2(\mathcal{N})$. These thresholds divide the distinguishability behavior into three distinct phases. 
At low noise rates where \(H(\mathcal{N})<0\), the ensemble lies in a \emph{resilient phase}: the average pairwise distance remains close to \(1\), indicating that the states stay nearly orthogonal. 
At high noise rates, once the collision-entropy threshold is crossed, \(H_2(\mathcal{N})\geq 0\), the ensemble enters a \emph{collapsed phase}, in which the distinguishability decays exponentially in the large-system limit. 
Between these two regimes is an \emph{intermediate phase} 
$H_2(\mathcal{N})<0\leq H(\mathcal{N})$, where the distinguishability drops from near unity to an algebraic scale \(O(1/\sqrt n)\) but has not yet undergone exponential collapse.

Together, these three phases capture the progressive degradation of state distinguishability as the noise rate increases. Physically, in the low-noise regime, the protective delocalization of logical information induced by scrambling dominates, effectively mitigating localized errors. However, as the noise rate increases and surpasses the threshold, the scrambling mechanism backfires by rapidly propagating errors across the entire system. This overwhelming non-local corruption suddenly destroys the protective effect, resulting in a complete loss of information.

Crucially, the existence of a noise-resilient phase preserves the feasibility of many quantum information processing tasks in the presence of noise.  For instance, randomly selected pairs of stabilizer states naturally enable robust quantum data hiding. Below the critical noise threshold, the global distinguishability of these pairs remains near unity, yet their distinguishability under local measurements is fundamentally suppressed due to the local thermalization of random stabilizer states. Furthermore, by elevating our average distance metric to high-probability concentration bounds, we prove that the resilient phase hosts mutually distinguishable subsets whose size scales doubly exponentially with the number of qubits. This robust subset secures the feasibility of quantum cryptographic protocols such as fingerprinting~\cite{Buhrman2001fingerprinting} in noisy environments. Notably, this distinguishability threshold precisely matches the quantum error correction threshold for random stabilizer codes~\cite{Hayden2008Capacity}. This deep connection emerges naturally in our lower bound proof: determining whether two random states from a 2-design remain distinguishable under noise is mathematically equivalent to asking whether a single logical qubit can be successfully protected by a randomized encoding.

The information extractable by a measurement $\mathcal{F}$ from a quantum state ensemble is governed by the average pairwise distance of the post-measured ensemble $\{p_{x},\mathcal{F}(\rho_{x})\}$. Evaluating this distinguishability for post-measured noisy ensembles reveals a behavior in drastic contrast to the unmeasured case. We find that for any measurement, post-measured noisy 2-design ensembles exhibit only exponentially small distinguishability under local purity-shrinking noise, excluding the existence of a positive noise threshold and any fault-tolerant phase. By duality, we prove this no-go result extends to scenarios involving arbitrary input states and 2-design measurements. This implies that quantum learning tasks relying on 2-designs, such as global classical shadow tomography~\cite{huang2020shadow}, inherently suffer from unbounded sample complexity unless actively protected by fault tolerance.

We summarize several key information-theoretic perspectives in Fig.~\ref{fig:main}. Our findings show that, with high probability, a unitary 2-design ensemble can maintain high distinguishability under noise in certain regimes, with its detailed behavior characterized by three distinct phases, but any measurement fundamentally changes the situation and destroys the distinguishability. This behavior suggests an interesting connection to the classical capacity of quantum channels~\cite{Hausladen1996capacity,Datta2013Capacity}, unifying both concepts under a common framework: transmitting classical messages via encoded quantum states through a noisy channel and applying state discrimination~\cite{Stephen2009discrimination} to decode messages. The study of capacity assumes optimal encoding and decoding to maximize transmission rates, realizing valid transmission only when the regularized Holevo capacity is positive~\cite{Schumacher1997classical,Holevo1998capacity,Datta2013Capacity}. In contrast, our average pairwise distance corresponds to communicating a single bit using two random states and an optimal measurement. Because this randomized encoding lacks optimal structure, its distinguishability threshold is more stringent than the maximum classical capacity, given by the single-letter coherent information. Furthermore, the post-measured ensemble corresponds to utilizing the entire 2-design as the encoding alphabet, where we prove no positive threshold exists. From an information-theoretic perspective, our approach captures the generic baseline performance of unstructured quantum protocols, establishing fundamental bounds for communication and learning without relying on sophisticated encodings.

\begin{figure*}[htbp!]
\centering
\includegraphics[width=.6\textwidth]{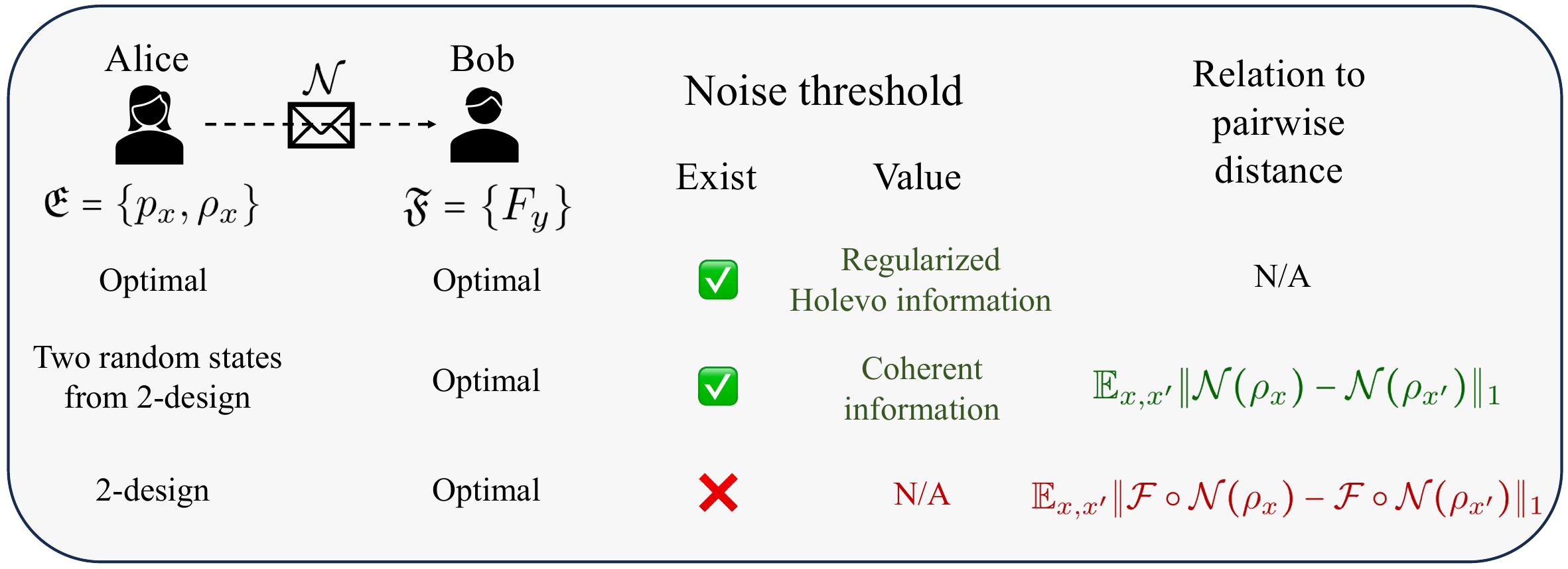}
\caption{Information-theoretic perspectives of our results. We unify the classical capacity of quantum channels and state distinguishability within an information-transmission framework. Optimal encoding and measurement yield the highest noise threshold, determined by the regularized Holevo information. Encoding with two random states from a 2-design features a more stringent threshold governed by coherent information. This corresponds to the quantum capacity of channel $\mathcal{N}$, which is inherently smaller than the classical capacity. Below this threshold, reliable information delivery remains possible. Conversely, encoding with the entire 2-design ensemble yields no positive threshold: the post-measured average pairwise distance and mutual information always decay exponentially under any local purity-shrinking noise.}
\label{fig:main}
\end{figure*}

\section{Preliminaries and ensemble distinguishability}
We summarize key notation below, with further mathematical details provided in Appendix~\ref{appendsc:pre}. We consider an $n$-qubit system $A$ with Hilbert space dimension $d=2^n$. Quantum states are represented by density matrices $\rho \in \mathcal{D}(\mathcal{H}_A)$ and are pure if and only if $\tr(\rho^2)=1$. $\ket{\psi}$ denotes a pure state  with corresponding density matrix $\psi = \ketbra{\psi}$. Evolutions of the quantum system are described by completely positive trace-preserving (CPTP) maps, or quantum channels. The channel $\mathcal{N}$ is fully characterized by its Choi state $\tau_{\mathcal{N}} = (\mathcal{N} \otimes \mathcal{I})(\ketbra{\Phi})$, where $\ket{\Phi} = \sum_i \ket{i}\ket{i}/\sqrt{d}$ is the maximally entangled state on two copies of systems, $A \otimes A'$. A channel $\mathcal{N}$ is said to be unital if it satisfies $\mathcal{N}(\id) = \id$. Measurement is described by positive operator-valued measure (POVM) $\mathfrak{F} = \{F_y\}$, satisfying $F_y \geq 0$ and $\sum_y F_y = \mathbb{I}$.

The celebrated Holevo--Helstrom bound establishes the optimal success probability for distinguishing two states $\rho$ and $\sigma$ (drawn with equal prior probability) as
\begin{equation}
p_{\text{succ}} = \frac{1}{2}\left[1 + \frac{1}{2}\Vert\rho - \sigma\Vert_1\right],
\end{equation}
where $\frac{1}{2}\Vert\rho - \sigma\Vert_1$ is the trace distance~\cite{HOLEVO1973Statistical,Helstrom1969estimation}. For a state ensemble $\mathfrak{E} = \{p_x, \rho_x\}$, we naturally generalize the definition and use average pairwise distance to characterize the distinguishability.

\begin{definition}[Ensemble distinguishability]
For an ensemble \(\mathfrak{E}=\{p_x,\rho_x\}\), we define its distinguishability as the average pairwise trace distance:
\begin{equation}
\mathbb{E}_{x,x'\sim \mathfrak{E}}
\left[
\frac{1}{2}\Vert\rho_x-\rho_{x'}\Vert_1
\right],
\end{equation}
where \(\rho_x\) and \(\rho_{x'}\) are independently drawn from the ensemble \(\mathfrak{E}\).
\end{definition}

This metric naturally characterizes the average optimal success probability for distinguishing two randomly drawn states. It also determines the capacity to transmit a single bit of information via a random two-state system. Furthermore, extending this average result to the high-probability case allows us to identify subsets of states exhibiting high mutual distinguishability.

This work focuses on the distinguishability of the noisy 2-design ensemble, $\{p_x, \mathcal{N}(\rho_x)\}$, where $\{p_x, \rho_x\}$ is a 2-design ensemble, and $\mathcal{N}$ is the noise channel. The state 2-design ensemble satisfies
\begin{equation}
\mathbb{E}_x[ \rho_x^{\otimes 2} ] = (\id+S)/(d(d+1)),
\end{equation}
where $\id$ and $S$ are identity and SWAP operators, respectively. Note that for the noiseless state 2-design ensemble, all components $\rho_x$ have to be pure. Similarly, the unitary 2-design ensemble $\mathfrak{U} = \{p_k, U_k\}$ satisfies
\begin{equation}
\begin{split}
\mathbb{E}_{U \sim \mathfrak{U}}[U^{\otimes 2} (\cdot) U^{\dagger\otimes 2}] &= \int_{\text{Haar}} dU\, U^{\otimes 2} (\cdot) U^{\dagger\otimes 2}\\
&= \tr(\cdot \frac{\id-d^{-1}S}{d^2-1} )\id + \tr(\cdot \frac{S-d^{-1}\id}{d^2-1} )S.
\end{split}
\end{equation}
Prominent examples include the multiqubit Clifford group and random stabilizer states as a unitary and a state 2-design, respectively. We also analyze the post-measured noisy ensemble, $\{p_x, \mathcal{F}\circ\mathcal{N}(\rho_x)\}$, where $\mathcal{F}(\rho) = \sum_y \tr F_y\rho \ketbra{y}$ is the measurement channel associated with POVM $\mathfrak{F} = \{F_y\}$.

Besides investigating the distinguishability, we frame these ensembles within an operational communication task: Alice selects an index $X$ from a distribution, prepares $N$ copies of $\rho_X$, and transmits them to Bob. Bob performs a POVM $\mathfrak{F}^i = \{F^i_y\}$ on the $i$-th copy, collects outcomes $Y = (Y_1, \dots, Y_N)$, and employs a predictor $\Bar{X} = f(Y)$ to estimate $X$. The protocol's performance is quantified by the mutual information $I(X:\bar{X})$ between the encoded and estimated messages. By data processing inequality, since $\Bar{X} = f(Y)$, we have that $I(X:\bar{X})\leq I(X:Y)$. We will demonstrate that this mutual information is intrinsically linked to the average pairwise distance of the state ensemble and the chosen POVM, as shown in the following theorem.


\begin{theorem}
\label{thm:pairwisedistance2mutualinformation}
Consider an $n$-qubit communication protocol where Alice encodes a classical variable $X$ into a state $\rho_X$, transmitted through a channel $\mathcal{N}$. Bob performs a measurement modeled by a quantum-to-classical channel $\mathcal{F}$ to obtain outcome $Y$. The mutual information $I(X:Y)$ satisfies:
\begin{equation}
\begin{split}
&I(X:Y)\leq \frac{n}{2} \mathbb{E}_{x,x^{\prime}}\Vert \mathcal{F}\circ\mathcal{N}(\rho_{x}) - \mathcal{F}\circ\mathcal{N}(\rho_{x^{\prime}})\Vert_1\\
&+ h_2\left(\min\{\frac{1}{2}, \frac{1}{2}\mathbb{E}_{x,x^{\prime}}\Vert \mathcal{F}\circ\mathcal{N}(\rho_{x}) - \mathcal{F}\circ\mathcal{N}(\rho_{x^{\prime}})\Vert_1\}\right);
\end{split}
\end{equation}
\begin{equation}
I(X:Y)\geq \frac{1}{8}
\left( \mathbb{E}_{x,x'}
\Vert \mathcal{F}\circ\mathcal{N}(\rho_{x}) - \mathcal{F}\circ\mathcal{N}(\rho_{x^{\prime}})\Vert_1
\right)^2.
\end{equation}
\end{theorem}

The proof, which is detailed in Appendix~\ref{appendssc:pairwisemulinf}, treats $\mathcal{F} \circ \mathcal{N}$ as an effective quantum-to-classical channel. The upper bound is derived by combining the Holevo bound with the Audenaert--Fannes inequality~\cite{audenaert2006sharp}, whereas the lower bound is obtained using Pinsker's inequality. Crucially, Theorem~\ref{thm:pairwisedistance2mutualinformation} demonstrates that the scaling of the mutual information $I(X:Y)$ is tightly governed by the average pairwise trace distance of the post-measured ensemble, with a maximum gap scaling linearly with $n$. In particular, if the average pairwise distance is exponentially small, the mutual information will have the same scaling. The average pairwise distance hence provides a quantitative measure of the performance of the communication protocol.

\section{Average pairwise distance under noise}
We now establish rigorous bounds on the average pairwise distance of a noisy state ensemble $\{p_x, \mathcal{N}(\rho_x)\}$, given by $\mathbb{E}_{x,x^{\prime}}\Vert\mathcal{N}(\rho_{x})-\mathcal{N}(\rho_{x^{\prime}})\Vert_{1}/2$. Evaluating this metric requires us to carefully track the information flow from the quantum system to the environment when encoded within a scrambled ensemble. To this end, we leverage the decoupling framework to characterize how this scrambling effect protects the encoded information.

\begin{theorem}\label{thm:noisydistinguish}
Suppose state set $\mathfrak{S} = \{U\ketbra{0^n}U^{\dagger}, U\in \mathfrak{U}\}$ is generated by unitary group $\mathfrak{U}$, then
\begin{equation}\label{eq:upper}
\mathbb{E}_{\psi_1,\psi_2\sim \mathfrak{S}}\Vert\mathcal{N}({\psi_1})-\mathcal{N}({\psi_2})\Vert_1 / 2 \leq \epsilon_{\mathcal{N}, \mathfrak{U}, \ketbra{0^n}};
\end{equation}
\begin{equation}\label{eq:lower}
\begin{split}
&\mathbb{E}_{\psi_1,\psi_2\sim \mathfrak{S}}\Vert\mathcal{N}({\psi_1})-\mathcal{N}({\psi_2})\Vert_1 / 2 \geq \mathbb{E}_{\psi_1,\psi_2\sim \mathfrak{S}}\Vert {\psi_1} - {\psi_2} \Vert_1 / 2\\
&- 4\sqrt{\max_W \epsilon_{\hat{\mathcal{N}}, \mathfrak{U}_W, \ketbra{\Phi}\otimes \ketbra{0^{n-1}}}}.
\end{split}
\end{equation}
Here, we denote $\mathfrak{U}_W = \{UW, U\in \mathfrak{U}\}$ while $W$ is an arbitrary unitary operation. Also,
\begin{equation}\label{eq:decouple}
\begin{split}
&\epsilon_{\mathcal{N}, \mathfrak{U}, \rho} \equiv\\
&\mathbb{E}_{U\sim \mathfrak{U}}\Vert \mathcal{N}_{S\rightarrow E}(U_S\rho_{SR} U_S^{\dagger}) - \mathbb{E}_{U\sim \mathfrak{U}}\mathcal{N}_{S\rightarrow E}(U_S\rho_{S} U_S^{\dagger})\otimes \rho_R\Vert_1
\end{split}
\end{equation}
is the decoupling error corresponding to channel $\mathcal{N}$, initial state $\rho$ and unitary group $\mathfrak{U}$.
\end{theorem}

Theorem~\ref{thm:noisydistinguish} rigorously reveals a profound physical connection between state distinguishability and quantum error correction induced by information scrambling. The upper bound in \eqref{eq:upper} captures the effect of state concentration: if the noise channel causes the output states to strongly concentrate around a single average state, the information of the input state is lost, and the distinguishability gets destroyed. Conversely, the lower bound in \eqref{eq:lower} reveals how a scrambled ensemble protects itself. By treating a random pair of states $\psi_1, \psi_2$ as the basis of a logical qubit, the loss of distinguishability is physically equivalent to the logical error induced by the environment. If the complementary channel $\hat{\mathcal{N}}$ limits the environment's ability to extract this logical information, which is characterized by a small decoupling error $\epsilon_{\hat{\mathcal{N}}}$, the encoded information is preserved~\cite{Beny2010AQEC,Liu2026AQEC}, and high distinguishability survives. The full proof is detailed in Appendix~\ref{appendsc:pairwise}.

We now apply this framework to the crucial model of 2-design ensembles. For unitary 2-designs, the decoupling error is dictated by the conditional collision entropy~\cite{Dupuis2014Decoupling,Szehr2013Decoupling}, defined for a bipartite state $\rho_{SE}$ as $H_2(S|E)_{\rho_{SE}} = -\log \tr[ ( (\id_S \otimes \rho_E^{-1/4}) \rho_{SE} (\id_S \otimes \rho_E^{-1/4}) )^2 ]$. We denote the channel entropy as $H_2(\mathcal{N}) = H_2(S|S')_{\mathcal{N}_{S'}(\ketbra{\Phi}_{SS'})}$, and its complementary counterpart as $H_2(\hat{\mathcal{N}}) = H_2(S|E)_{\hat{\mathcal{N}}_{S'\rightarrow E}(\ketbra{\Phi}_{SS'})}$. The quantity $H_2^{\delta}(\rho)$ refers to the smooth conditional collision entropy, optimized over states $\delta$-close to $\rho$ in purified distance. Combining these entropic quantities with Theorem~\ref{thm:noisydistinguish} yields the following threshold behaviors.

\begin{theorem}\label{thm:2design}
Suppose $\mathfrak{S}$ is generated by a unitary 2-design ensemble, then
\begin{equation}\label{eq:2designupperbound}
\mathbb{E}_{\psi_1,\psi_2\sim \mathfrak{S}}\Vert\mathcal{N}(\psi_1)-\mathcal{N}(\psi_2)\Vert_1/2 \leq\begin{cases}
2^{-\frac{1}{2}H_2(\mathcal{N})}, \\
2^{-\frac{1}{2}H^{\delta}_2(\mathcal{N})}+4\delta.
\end{cases};
\end{equation}
\begin{equation}
\begin{split}\label{eq:2designlowerbound}
\mathbb{E}_{\psi_1,\psi_2\sim \mathfrak{S}}&\Vert\mathcal{N}(\psi_1)-\mathcal{N}(\psi_2)\Vert_1/2 \geq\\
&\begin{cases}
1-2^{-n}-2^{\frac{9}{4}-\frac{1}{4}H_2(\hat{\mathcal{N}})}, \\
1-2^{-n}-4\sqrt{2^{\frac{1}{2}-\frac{1}{2}H_2^{\delta}(\hat{\mathcal{N}})} + 4\delta}.
\end{cases},
\end{split}
\end{equation}
where $H_2$ and $H_2^\delta$ are the conditional collision entropy and smooth conditional collision entropy, respectively; $\delta$ is the smoothing parameter.
\end{theorem}

Note that the upper bound \eqref{eq:2designupperbound} strictly holds for any state 2-design, while the lower bound \eqref{eq:2designlowerbound} requires the unitary 2-design condition, as its derivation relies on the explicit error-correcting properties of random Clifford encodings.

Theorem~\ref{thm:2design} reveals a remarkable physical phenomenon: the average distinguishability of scrambled ensembles does not degrade smoothly, but rather undergoes sharp phase transitions dictated by the channel's entropies. For local noise of the form $\mathcal{N}^{\otimes n}$, where macroscopic scaling follows $H_2^{\delta}(\mathcal{N}^{\otimes n}) \geq n(H(\mathcal{N})-\delta)$, we can rigorously classify the survival of the ensemble into three distinct phases:

\begin{enumerate}
\item
\textbf{Resilient Phase} ($H(\mathcal{N}) < 0$): In this low-noise regime, scrambling effectively protects the encoded information. The average pairwise distance remains close to $1$, meaning the randomized states resist the noise and stay macroscopically distinguishable.
\item 
\textbf{Intermediate Phase} ($H_2(\mathcal{N}) < 0 \leq H(\mathcal{N})$): This regime is a transition region arising from the gap between the von Neumann entropy $H$, which governs the average information limit, and the collision entropy $H_2$, which dictates strict state concentration. Because a 2-design only mimics true randomness up to the second moment, it is susceptible to larger statistical fluctuations. In this window, the noise is strong enough to destroy the near-perfect distinguishability ($H \geq 0$), but not yet strong enough to completely concentrate the output states into a single indistinguishable average ($H_2 < 0$). Consequently, the distinguishability drops from unity to an algebraic decay of $O(1/\sqrt{n})$.
\item 
\textbf{Collapsed Phase} ($H_2(\mathcal{N}) \geq 0$): Once the noise crosses the collision entropy threshold, the protective effect of scrambling is overwhelmed by its error-spreading effect. Localized errors rapidly propagate across the entire system, forcing the output states to strongly concentrate around the maximally mixed state. The distinguishability decays exponentially, $\exp(-\Omega(n))$, leading to the loss of information.
\end{enumerate}

The resilient phase corresponds to \eqref{eq:2designlowerbound}, and the collapsed and intermediate phases correspond to the first and second inequalities of \eqref{eq:2designupperbound}, respectively. Mathematically, this intermediate window arises from the gap between the von Neumann entropy $H(\mathcal{N})$ and the collision entropy $H_2(\mathcal{N})$, as well as the $O(1/\sqrt{n})$ scaling of the smoothing parameter $\delta$ required to bridge them via the quantum asymptotic equipartition property~\cite{Tomamichel2009QAEP}. While this intermediate regime is distinct from the resilient phase, its exact boundary with the exponential collapsed phase remains an intriguing open question. Resolving this will require analytical techniques that extend beyond the standard 2-design decoupling framework, for instance, developing non-smooth decoupling theorems utilizing R{\'e}nyi entropies $H_{\alpha}$ with $\alpha \to 1$. 
Understanding this boundary more precisely is an interesting direction for future work and would lead to a more complete picture of the phases of noisy scrambled ensembles.

To illustrate this abstract phase diagram, consider a local depolarizing noise channel $\mathcal{N}_p(\rho) = (1-p)\rho+p\frac{\id}{2}$. The critical threshold $H(\mathcal{N}) = 0$ is crossed at a physical noise rate of $p \approx 0.2524$, marking the exit from the resilient phase, while the terminal $H_2(\mathcal{N}) = 0$ threshold occurs at $p \approx 0.4227$. The behavior for local Pauli noise exhibits identical structural transitions, as illustrated in Fig.~\ref{fig:pauliphase}.

\begin{figure}[htbp!]
\centering
\includegraphics[width=.45\textwidth]{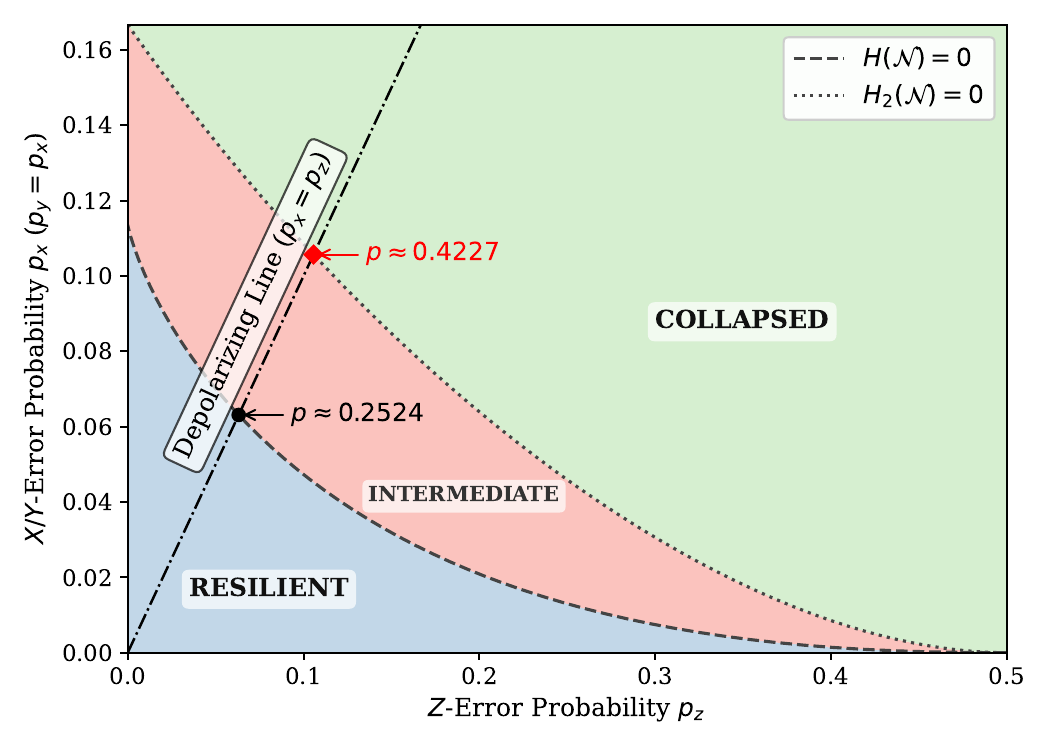}
\caption{Distinguishability phase diagram for local Pauli noise $\mathcal{N}(\rho) = (1-p_x-p_y-p_z)\rho + p_x X\rho X + p_y Y\rho Y + p_z Z\rho Z$, assuming $p_x = p_y$. The zero-entropy bounds of the channel, $H(\mathcal{N}) = 0$ and $H_2(\mathcal{N}) = 0$, sharply partition the noise space into resilient, intermediate, and collapsed phases. The local depolarizing noise trajectory ($p_x = p_z$) intersects these phase transitions at threshold parameters $p \approx 0.2524$ and $p \approx 0.4227$, which correspond to the plotted axis values of $p/4$.}
\label{fig:pauliphase}
\end{figure}

While the unmeasured quantum ensemble exhibits a noise-resilient phase, extracting this information via classical measurement presents a drastically different picture. According to Theorem~\ref{thm:pairwisedistance2mutualinformation}, we apply Theorem~\ref{thm:2design} to the post-measured ensemble $\{p_x, \mathcal{F}\circ \mathcal{N}(\rho_x)\}$ and reveal that the protective power of the 2-design completely vanishes. The upper bound of the average pairwise distance is strictly constrained by $\sqrt{\Vert \mathcal{N}^{\otimes 2}(S)\Vert_{\infty}}$ for any measurement and unital channel $\mathcal{N}$, where $S$ denotes the swap operator on two copies of the system and $\Vert \cdot \Vert_{\infty}$ is the spectral norm. Measurement destroys the quantum information, making coherent quantum distinguishability degrade to mixed classical distinguishability with no error tolerance. Physically, this quantity measures how severely the noise channel shrinks the overlap between two copies of the system. For any local purity-shrinking noise, such as depolarizing noise, this bound decays exponentially with the number of qubits $n$. Consequently, the extractable mutual information vanishes exponentially, definitively precluding the existence of any finite noise threshold.

By duality, this severe loss of information equally applies to scenarios where the measurement itself, instead of the input states, constitutes a 2-design. The only change is to replace $\mathcal{N}$ with its adjoint map $\mathcal{N}^{\dagger}$.

Besides using average pairwise distance to provide bounds on mutual information, we further derive tighter upper bounds by leveraging the $\chi^2$-divergence to bound the Kullback-Leibler divergence. See Appendices~\ref{appendssc:mulinf2design} and~\ref{appendssc:mulinf2designPOVM} for full mathematical details. Finally, we obtain that the mutual information is bounded by
\begin{equation}\label{eq:mutualinfboundstate}
\Vert \mathcal{N}^{\otimes 2}(S)\Vert_{\infty},
\end{equation}
if the state ensemble forms a noisy 2-design, and is bounded by
\begin{equation}\label{eq:mutualinfboundPOVM}
\Vert (\mathcal{N}^{\dagger})^{\otimes 2}(S)\Vert_{\infty},
\end{equation}
if the POVM elements form a noisy 2-design. Ultimately, regardless of whether the scrambling lies in the state preparation or the measurement basis, unmitigated local noise fundamentally obstructs the extraction of classical information from the ensemble.

\section{Applications}
We now demonstrate the operational implications for quantum information processing tasks. Our results establish fundamental boundaries for these tasks by revealing a contrast in two cases of state distinguishability. Specifically, unmeasured state ensembles exhibit three distinct phases characterized by a strictly positive noise threshold. Conversely, post-measured ensembles suffer an immediate distinguishability collapse and exhibit no noise threshold. We explain various practical applications of these findings below.

First, the noise-resilient phase demonstrates that the inherent scrambling of 2-designs natively protects randomly selected state pairs from noise. Because stabilizer states form a 2-design~\cite{Webb2016Clifford3design}, utilizing shared classical randomness to select a random stabilizer pair yields a highly robust, unstructured protocol for reliable classical communication. Furthermore, because the global distinguishability of these pairs survives while their distinguishability under restricted local measurements is fundamentally suppressed, this phase intrinsically enables robust quantum data hiding. It also guarantees the noise-resilience of advanced quantum cryptographic protocols like quantum fingerprinting. We detail the cryptography applications in Sections~\ref{ssc:datahiding} and \ref{ssc:fingerprinting}.

Conversely, our post-measurement bounds impose a strong limitation on classical information extraction from noisy scrambled ensembles. Protocols relying on unstructured scrambled measurements, such as classical shadow tomography, can therefore suffer a fundamental sample complexity blowup under noise. 
 
Under such a measurement, the scrambling-induced protection collapses, and the resilient phase disappears. 
We discuss these limitations for quantum learning in Section~\ref{ssc:learning}.
This contrasts sharply with the unmeasured case where each target pair utilizes a tailored measurement for distinguishability and information extraction. When all scrambled states within the ensemble are subjected to a single universal measurement, the scrambled protection collapses and the resilient phase completely disappears. We discuss these limitations for quantum learning in Section~\ref{ssc:learning}.

\subsection{Noise-resilient phase and robust data hiding}\label{ssc:datahiding}
Many cryptographic protocols are built upon exploiting the operational gap between global and restricted measurements. A prominent example of this broader class of tasks is quantum data hiding~\cite{DiVincenzo2002hiding,Hayden2004Randomizing,Matthews2009Distinguishability}, which aims to encode a classical message into a multipartite quantum state such that parties with joint global measurements can perfectly decode it, while parties restricted to local measurements can extract virtually no information~\cite{DiVincenzo2002hiding}. At their core, these tasks strictly require codebooks that exhibit a massive distinguishability gap: their global trace distance must be close to 1 to allow perfect decoding, while their local trace distance must vanish to ensure security.

For a noiseless $n$-qubit system, randomly selected pairs of stabilizer states natively achieve this separation. Since the average overlap of a 2-design is $2^{-n}$, their average pairwise distance approaches the theoretical maximum, 
\begin{equation}
\mathbb{E}_{\psi_1, \psi_2} \Vert \psi_1 - \psi_2 \Vert_1 / 2 \geq 1 - 2^{-n}.
\end{equation}
This ensures near-perfect global decoding. Conversely, to evaluate the distinguishability under restricted local measurements, consider an observer accessing only a subsystem $A$ of $n_A$ qubits, while tracing out the remaining $n_{\Bar{A}} = n-n_A$ qubits. Because random stabilizer states mimic the entanglement spectrum of Haar-random states up to the second moment, we invoke the decoupling theorem~\cite{Hayden2008Capacity}. The local marginals exponentially thermalize to the maximally mixed state $\id_A / d_A$, bounding the local distinguishing power via the triangle inequality:
\begin{equation}
\mathbb{E}_{\psi_1, \psi_2} \Vert \rho_A^{(1)} - \rho_A^{(2)} \Vert_1  \leq 2^{\frac{2n_A - n}{2}},
\end{equation}
where $\rho_A^{(i)} = \tr_{\Bar{A}} \psi_i$. Thus, as long as the inaccessible subsystem is much larger than half, the unstructured scrambling intrinsically hides the encoded data from local observers.

A critical question is whether this cryptographic distinguishability gap survives in the presence of realistic noise. In typical protocols, environmental noise degrades the global coherences required for decoding, rapidly destroying the protocol. However, our results guarantee that for local noise $\mathcal{N}^{\otimes n}$ strictly within the resilient phase ($H(\mathcal{N}) < 0$), the average pairwise distance remains lower-bounded by $1 - o(1)$. Simultaneously, by the data processing inequality, the local distinguishability for the restricted eavesdropper can only decrease further under noise.

Consequently, the existence of the noise-resilient phase fundamentally secures this broad class of cryptographic tasks. While prior work established the theoretical capacity for noisy quantum data hiding using asymptotic Haar-random block codes~\cite{Lupo2016datahiding}, our findings prove that unstructured 2-designs natively achieve this protection without complex encodings. Unmitigated local noise cannot blind the global receiver, nor can it break the secrecy against local observers. This proves that random stabilizer pairs serve as robust quantum codebooks, naturally preserving the data hiding capacity and related distinguishability gaps without requiring active error correction.

\subsection{Concentration of distinguishability and quantum cryptography}\label{ssc:fingerprinting}
Tasks such as quantum fingerprinting and digital signatures require a codebook of quantum states where every pair is mutually distinguishable, allowing each state to reliably encode a distinct message. In the noiseless regime, the optimal size of such a mutually distinguishable ensemble scales doubly exponentially with the number of qubits $n$~\cite{Buhrman2001fingerprinting}. A critical open question is whether this scaling survives in the presence of noise, and if so, what the operational noise thresholds are.

To answer this, we elevate our average pairwise distance bounds to high-probability statements. While a naive deviation bound could be obtained via Chebyshev's inequality, we achieve a tight bound by considering Haar-random states and leveraging the concentration of measure in high-dimensional spaces. The distance function $f(\psi_1, \psi_2) = \Vert\mathcal{N}(\psi_1) - \mathcal{N}(\psi_2)\Vert_1$ is 2-Lipschitz in both inputs, or $\abs{f(\psi_1, \psi_2)-f(\phi_1, \phi_2)}\leq 2\Vert \ket{\psi_1}-\ket{\phi_1}\Vert_2 + 2\Vert \ket{\psi_2}-\ket{\phi_2}\Vert_2$. Using Lévy's lemma~\cite{ledoux2001concentration}, we show that the trace distance concentrates exponentially fast in the Hilbert space dimension $d=2^n$, satisfying $\Pr_{\psi_1, \psi_2} ( |f(\psi_1, \psi_2) - \mathbb{E}[f]| \geq \epsilon ) \leq 4 \exp( -\frac{d \epsilon^2}{36\pi^3} )$.

\begin{corollary}\label{coro:haarhighprob}
For an ensemble $\mathfrak{S}$ generated by Haar-random states, the trace distance concentrates as:
\begin{equation}
\begin{split}
\Pr_{\psi_1,\psi_2\sim \mathfrak{S}}&(\Vert\mathcal{N}(\psi_1)-\mathcal{N}(\psi_2)\Vert_1 / 2 \geq \epsilon) \leq\\
&\begin{cases}
4 \exp( -\frac{2^n (\epsilon-2^{-\frac{1}{2}H_2(\mathcal{N})})^2}{36\pi^3} ), \\
4 \exp( -\frac{2^n (\epsilon-2^{-\frac{1}{2}H^{\delta}_2(\mathcal{N})}-4\delta)^2}{36\pi^3} ).
\end{cases};\\
\end{split}
\end{equation}
\begin{equation}
\begin{split}
\Pr_{\psi_1,\psi_2\sim \mathfrak{S}}&(\Vert\mathcal{N}(\psi_1)-\mathcal{N}(\psi_2)\Vert_1 / 2 \leq 1-\epsilon) \leq\\
&\begin{cases}
4 \exp( -\frac{2^n (\epsilon-2^{-n}-2^{\frac{9}{4}-\frac{1}{4}H_2(\hat{\mathcal{N}})})^2}{36\pi^3} ), \\
4 \exp( -\frac{2^n (\epsilon-2^{-n}-4\sqrt{2^{\frac{1}{2}-\frac{1}{2}H_2^{\delta}(\hat{\mathcal{N}})} + 4\delta})^2}{36\pi^3} ).
\end{cases}.
\end{split}
\end{equation}
\end{corollary}

Applying the union bound guarantees the existence of a subset of states $\mathcal{S} \subset \mathfrak{S}$ of size $M$ such that every distinct pair in $\mathcal{S}$ is highly distinguishable. Specifically, we can construct a subset of size $M \sim \exp(c 2^n (\epsilon-2^{-n}-4\sqrt{2^{\frac{1}{2}-\frac{1}{2}H_2^{\delta}(\hat{\mathcal{N}})} + 4\delta})^2)$ where all state pairs have a trace distance strictly greater than $1-\epsilon$, for some constant $c$.

Crucially, for local noise $\mathcal{N}^{\otimes n}$ strictly within the resilient phase ($H(\mathcal{N}) < 0$), we can choose $\epsilon$ to be an arbitrarily small constant while maintaining $M \sim \exp(c' 2^n)$ for some constant $c'$. This proves that quantum fingerprinting protocols remain viable in realistic, noisy environments: provided the noise is below the entropic threshold defined by the coherent information $H(\mathcal{N})$, the distinguishable set remains doubly exponentially large and no active error correction is required. This scaling perfectly matches the noiseless optimal limit~\cite{Buhrman2001fingerprinting}.

\subsection{Post-measured ensemble distinguishability and quantum learning}\label{ssc:learning}
Quantum learning protocols can generally be mapped to information transmission tasks. Consider a protocol learning a specific state from a uniformly distributed ensemble of size $|X|$. The learner receives $N$ independent copies of the state, performs single-copy measurements, and outputs a guess $\Bar{X}$ with failure probability $\delta = \Pr(X \neq \Bar{X})$, which normally decays with $n$. By Fano's inequality, the mutual information is bounded by:
\begin{equation}
I(X:\Bar{X})\geq \log\abs{X} - \delta \log(\abs{X}-1)-h(\delta)\approx \log\abs{X},
\end{equation}
where the approximation holds in the large system limit and small failure probability. By the subadditivity, the $N$ copies can provide at most $N I(X:Y)$ mutual information, the sample complexity is fundamentally bounded by:
\begin{equation}
N  = \Omega\left(\frac{\log\abs{X}}{I(X:Y)}\right).
\end{equation}
Thus, the sample complexity scales inversely with the single-copy mutual information $I(X:Y)$. Combining this with Eqs.~\eqref{eq:mutualinfboundstate} and~\eqref{eq:mutualinfboundPOVM}, the bounds on mutual information, yields the following requirement:

\begin{corollary}\label{coro:mutualinfbound}
For any quantum state learning protocol corrupted by noise $\mathcal{N}$, the sample complexity obeys the following lower bounds:

\begin{enumerate}
    \item If the state ensemble forms a \(2\)-design, then
    \begin{equation}
        N
        =
        \Omega\!\left(
        \frac{\log |X|}
        {\Vert\mathcal{N}^{\otimes 2}(S)\Vert_{\infty}
        }
        \right).
    \end{equation}

    \item If the POVM elements form a \(2\)-design, then
    \begin{equation}
        N
        =
        \Omega\!\left(
        \frac{\log |X|}
        {\Vert\mathcal{N}^{\dagger\otimes 2}(S)\Vert_{\infty}}
        \right).
    \end{equation}
\end{enumerate}
\end{corollary}

For unital and local purity-shrinking noise like depolarizing noise, $\Vert \mathcal{N}^{\otimes 2}(S)\Vert_{\infty}=\Vert \mathcal{N}^{\dagger\otimes 2}(S)\Vert_{\infty}$ is exponentially small in $n$. Consequently, the learning protocol necessarily demands an exponentially large sample complexity, rendering it inefficient without fault tolerance.

The implications of the above analysis extend beyond simple state identification to more general protocols, such as classical shadow tomography~\cite{huang2020shadow}. In this framework, a learner aims to predict linear observables within an accuracy $\epsilon$ using random Clifford measurements. If a shadow protocol succeeds, it can be systematically reduced to the aforementioned state-identification task~\cite{huang2020shadow}. By constructing $M$ distinguishable states $\rho_i$ and corresponding observables $O_i$ such that $\tr(O_i\rho_i) \geq \max_{j\neq i}\tr(O_j\rho_i) + 3\epsilon$, determining the observables reduces exactly to identifying the chosen index $i$. Applying our previous bound to this specific measurement structure yields:
\begin{corollary}\label{coro:shadowbound}
Classical shadow tomography utilizing a 2-design POVM under noise $\mathcal{N}$ requires a sample complexity of $N = \Omega\left( \Vert (\mathcal{N}^{\dagger})^{\otimes 2}(S)\Vert_{\infty}^{-1} \right)$.
\end{corollary}
This guarantees that for local purity-shrinking noise, classical shadows via 2-design POVMs suffer from an exponential sample complexity in $n$.

Furthermore, we can tighten this bound by exploiting the specific operational structure of shadow tomography. The shadow protocol is universal: for any unitary $U$, measuring $U\rho U^{\dagger}$ and knowing $U$ allows one to logically deduce the observables of $\rho$. By introducing random unitaries to bound the mutual information, this universal property yields an even tighter sample complexity lower bound of $\Omega((\tr(\tau_{\mathcal{N}}^2))^{-1})$, characterized directly by the purity of the Choi state of the noise channel.

This tightened bound is critical for analyzing asymmetric noise. For instance, under local dephasing noise, the previous bound $\Vert (\mathcal{N}^{\dagger})^{\otimes 2}(S)\Vert_{\infty}$ remains constant, which fails to capture the protocol's physical degradation. However, the purity $\tr(\tau_{\mathcal{N}}^2)$ becomes exponentially small, revealing that dephasing alone is sufficient to destroy the efficiency of classical shadow tomography. The full details are provided in Appendix~\ref{appendsc:universal}.

\section{Conclusion and Outlook}
To summarize, we introduced the notion of average pairwise distance to rigorously characterize quantum state ensemble distinguishability under noise, and derived tight information-theoretic bounds via the decoupling framework. These bounds reveal sharp phase transitions for noisy scrambled ensembles in distinguishability governed by conditional entropy thresholds. Operationally, our results guarantee robust data hiding and quantum cryptography tasks within the resilient phase. Conversely, the vanishing distinguishability of post-measured ensembles dictates an exponential sample complexity for quantum learning protocols like classical shadow tomography, precluding any positive constant noise thresholds.

There are several important directions for future work. An important theoretical question is to resolve the precise structure of the intermediate distinguishability phase. Determining whether this algebraic decay regime represents a fundamental physical boundary or is eventually absorbed into the exponential collapsed phase will require analytical techniques that go beyond standard 2-design decoupling properties.

Extending these results to approximate 2-designs is also important.  While the decoupling approach guarantees that the general structure of the resilient and intermediate phases persists for approximate designs, the distinguishability behavior in the collapsed phase is inherently bounded by the approximation error. For instance, in approximate 2-designs generated by one-dimensional logarithmic-depth circuits~\cite{schuster2024randomunitariesextremelylow,laracuente2024approximateunitarykdesignsshallow,Liu2026AQEC}, the approximation error scales as an inverse polynomial, suggesting that the exponential collapse is replaced by an inverse polynomial floor. Sharpening these bounds to determine if a true exponential collapse can still emerge under such relaxed conditions would be valuable.

Furthermore, while the average pairwise distance under noise has been well characterized, the finer structure of this distinguishability deserves further study. Specifically, the current results do not reveal which pairs of states maintain distinguishability under noise and which undergo a non-threshold collapse. Resolving this would naturally classify states within a 2-design into equivalence classes, establishing robust subsets for information transmission and learning. It is particularly compelling to investigate the relationship between this classification and circuit complexity. Because states distinguishable via local observables generally exhibit robustness against local noise, pairs uniquely vulnerable to noise-induced indistinguishability must have long-range entanglement. Consequently, these equivalence classes may be deeply connected to the structures of long-range entanglement and circuit complexity~\cite{Chen2010topological}.

From a practical perspective, an important challenge is to turn our randomized existence results into explicit deterministic constructions of doubly exponentially large sets of mutually distinguishable states, which would be useful for robust quantum fingerprinting and digital signatures. One promising strategy is to leverage good quantum error-correcting codes and utilize their logical states as the codebook. Finally, expanding our mathematical framework to bound the distinguishability of post-measured ensembles under non-unital noise, such as amplitude damping, is desirable for characterizing realistic experimental environments.

\begin{acknowledgements}
G.L., C.Q., and X.M.\ are supported by the National Natural Science Foundation of China (Grant No.~12575023), the Quantum Science and Technology-National Science and Technology Major Project (Grants No.~2021ZD0300804 and No.~2021ZD0300702), and the CCF-QuantumCtek Superconducting Quantum Computing Special Cooperation Program (Grant No.~CCF-QC2025005). G.L.\ acknowledges additional support by the Turing AI Institute of Nanjing (Grant No.~TR-IIIS 002). Z.X.\ acknowledges funding and support from Joint Center for Quantum Information and Computer Science (QuICS) Lanczos Graduate Fellowship. Z.-W.L.\ is supported in part by NSFC under Grant No.~12475023, Dushi Program, and startup funding from YMSC.
\end{acknowledgements}



\bibliography{bibNoisyHypothesis}

\onecolumngrid

\appendix
\section{Preliminaries}\label{appendsc:pre}
\subsection{Basic notations and quantities}\label{appendssc:notation}
Let us start with the basic notations used in this work. We denote the Hilbert space as $\mathcal{H}$, the set of the states on $\mathcal{H}$ as $\mathcal{D}(\mathcal{H})$, the set of the general matrices on $\mathcal{H}$ as $\mathcal{M}(\mathcal{H})$. The qubit Hilbert space is denoted with $\mathcal{H}_2$ with basis $\{\ket{0}, \ket{1}\}$. The Hilbert space of the $n$-qubit system is $\mathcal{H}_2^{\otimes n}$. The Hilbert space $\mathcal{H}_{SE}$ of the combined system $SE$ of subsystems $S$ and $E$ is $\mathcal{H}_{SE} = \mathcal{H}_S\otimes \mathcal{H}_E$. The dimension of $\mathcal{H}_S$ is denoted as $\abs{S}$. For a quantum state $\rho_{SE}\in \mathcal{D}(\mathcal{H}_{SE})$, its reduced density matrix on system $E$ is denoted as $\rho_E = \tr_S(\rho_{SE})$, where $\tr_S$ is the partial trace over subsystem $S$. The unitary operation $U$ on system $S$ is denoted as $U_S$ with the subscript representing the support of $U$. A completely positive and trace-preserving (CPTP) map $\mathcal{N}$ that maps a state from $\mathcal{H}_S$ to $\mathcal{H}_E$ is denoted as  $\mathcal{N}_{S \rightarrow E}$.
The Choi-Jamio{\l}kowski representation of a CPTP map $\mathcal{N}_{S\rightarrow E}$ is a quantum state, $\tau_{\mathcal{N}} = \mathcal{N}_{S'\rightarrow E}(\ketbra{\Phi}_{SS'})$, where $S'$ is a copy of system $S$, and $\ket{\Phi}_{SS'} = \frac{1}{\sqrt{\abs{S}}}\sum_{i=0}^{\abs{S}-1}\ket{i}_S\ket{i}_{S'}$ is the maximally entangled state on system $SS'$.

We use several measures to quantify the distance between two states. The fidelity is defined as
\begin{equation}
F(\rho, \sigma) = \Vert \sqrt{\rho}\sqrt{\sigma}\Vert_1 = \tr\sqrt{\sqrt{\rho}\sigma\sqrt{\rho}}.
\end{equation}
For pure state $\ket{\psi}$ and mixed state $\sigma$, the fidelity is given by
\begin{equation}\label{eq:fidelity_pure_mixed}
F(\ketbra{\psi}, \sigma) = \sqrt{\bra{\psi} \sigma \ket{\psi}}.
\end{equation}
From the fidelity, we can define the purified distance:
\begin{equation}
P(\rho,\sigma) = \sqrt{1-F(\rho, \sigma)^2},
\end{equation}
which satisfies~\cite{nielsen2002quantum}
\begin{equation}\label{eq:tracedistanceineq}
\frac{1}{2}\Vert \rho - \sigma \Vert_1 \leq P(\rho, \sigma) \leq \sqrt{\Vert \rho - \sigma\Vert_1}.
\end{equation}

The trace distance between $\rho$ and $\sigma$ is $\frac{1}{2}\Vert \rho - \sigma \Vert_1$ . The definitions of fidelity $F$ and trace distance can also be generalized to arbitrary positive matrices beyond states by substituting $\rho$ and $\sigma$ with general matrices. For two quantum channels $\mathcal{N}_{S\rightarrow E}$ and $\mathcal{N}'_{S\rightarrow E}$, we normally use diamond distance to characterize their distance, defined as
\begin{equation}
\Vert \mathcal{N}_{S\rightarrow E} - \mathcal{N}'_{S\rightarrow E} \Vert_{\diamond} = \max_{\rho_{SR}} \Vert \mathcal{N}_{S\rightarrow E}(\rho_{SR}) - \mathcal{N}'_{S\rightarrow E}(\rho_{SR})  \Vert_1,
\end{equation}
where $R$ is a reference system without dimension restriction. Nonetheless, it can be shown that considering a reference system $R$ with the same dimension as $S$ is sufficient for maximization. The definition of the diamond distance can also be generalized for two completely positive but non-trace-preserving maps.

In this work, we will massively use the conditional entropy. The conditional collision entropy of $S$ given $E$ for $\rho_{SE}$ is defined as
\begin{equation}
H_2(S|E)_{\rho} = -\log \tr[\left((\id_S\otimes \rho_E)^{-1/4}\rho_{SE}(\id_S\otimes \rho_E)^{-1/4}\right)^2].
\end{equation}
Its smoothed version is
\begin{equation}
H^{\delta}_{2}(S|E)_{\rho} = \sup_{\hat{\rho}_{SE}\in \mathfrak{E}^{\delta}(\rho_{SE})}H_{2}(S|E)_{\hat{\rho}},
\end{equation}
where $\mathfrak{E}^{\delta}(\rho) = \{\rho', P(\rho', \rho)\leq \delta\}$ is the set of all $\delta$-close states of $\rho$. Particularly, we will denote $H_2(S|E)_{\mathcal{N}_{S'\rightarrow E}(\ketbra{\Phi}_{SS'})}$ with $H_2(S|E)_{\mathcal{N}_{S\rightarrow E}}$, $H_2(S|E)_{\mathcal{N}}$ or simply $H_2(\mathcal{N})$. The same holds for $H^{\delta}_2$ and other conditional entropies.

Another useful entropy is the conditional von Neumann entropy, which is defined as
\begin{equation}
H(S|E) = H(SE)-H(E),
\end{equation}
where $H$ is the von Neumann entropy with $H(\rho) = -\tr\rho\log \rho$. One can also define the mutual information between $S$ and $E$,
\begin{equation}
I(S:E)  = H(S)-H(S|E) = H(E)+H(S)-H(SE),
\end{equation}
which is always non-negative. All entropies reduce to their classical counterpart when the quantum state is a mixture of computational-basis states.

The conditional collision entropy can converge to the conditional von Neumann entropy as follows.
\begin{lemma}[AEP~\cite{Tomamichel2009QAEP,Liu2026AQEC}]\label{lem:AEP} For quantum state $\rho_{SE}$, we have
\begin{equation}
H(S|E)_{\rho} - \frac{4\log \gamma \sqrt{\log \frac{2}{\delta^2}}}{\sqrt{n}} \leq \frac{1}{n}H^{\delta}_{2}(S|E)_{\rho^{\otimes n}} \leq H(S|E)_{\rho} + \frac{4\log \gamma \sqrt{\log \frac{2}{\delta^2}}}{\sqrt{n}},
\end{equation}
with $\gamma = O(1)$. When $\rho$ is a two-qudit state with local dimension $q$, $\gamma\leq 2\sqrt{q}+1$.
\end{lemma}

One important result concerning mutual information is Fano's inequality. For a communication protocol where Alice sends states to Bob, and Bob performs a measurement to get the result $Y$ and predict $X$ as $\Bar{X}$. We have that
\begin{equation}
H(X) - I(X:\bar{X}) \leq h(p_e) + p_e \log(|X|-1),
\end{equation}
where $p_e = \Pr(X \neq \bar{X})$ is the error probability. For a uniform prior $p_x$, this yields the lower bound $I(X:\bar{X}) \geq h(p_{\text{succ}}) + p_{\text{succ}} \log(|X|-1)$ where $p_{\text{succ}} = 1 - p_e$. Consequently, in the communication protocol, to achieve a constant success probability, the mutual information $I(X:\bar{X})$ must scale as $\log(|X|)$. Conversely, if $I(X:\bar{X})$ is exponentially small, the success probability will also be exponentially small. Therefore, investigating the scaling of the mutual information provides an accurate characterization of the success probability or the performance of the communication protocol.



According to Stinespring dilation, a channel $\mathcal{N}_{S}$ can be modeled by interacting the physical system $S$ with the environment $E$ through a unitary $U$, followed by tracing out the environment:
\begin{equation}
\mathcal{N}_S(\rho_S) = \tr_E(U_{SE}\rho_S\otimes \ketbra{0}_EU_{SE}^{\dagger}).
\end{equation}
If we trace out the system instead of the environment, we obtain the complementary channel of $\mathcal{N}_S$~\cite{Beny2010AQEC}:
\begin{equation}
\hat{\mathcal{N}}_{S\rightarrow E} = \tr_S(U_{SE}\rho_S\otimes \ketbra{0}_EU_{SE}^{\dagger}).
\end{equation}
Note that the complementary channel of $\mathcal{N}_S\circ \mathcal{U}_S$ with $\mathcal{U}_S$ a unitary operation can be obtained by
\begin{equation}
\widehat{\mathcal{N}\circ \mathcal{U}}_{S\rightarrow E} = \hat{\mathcal{N}}_{S\rightarrow E}\circ\mathcal{U}_{S}.
\end{equation}
Given the dual relationship between the complementary channel $\hat{\mathcal{N}}_{S\rightarrow E}$ and the channel $\mathcal{N}_S$, their conditional von Neumann entropy is opposite~\cite{Tomamichel2014renyi}:
\begin{equation}
H(S|E)_{\mathcal{N}_{S'\rightarrow E}(\ketbra{\Phi}_{SS'})} = - H(S|S')_{\mathcal{N}_{S'}(\ketbra{\Phi}_{SS'})}, 
\end{equation}
and similarly, we have the relation
\begin{equation}\label{eq:dualcollisionentropy}
H_2(S|E)_{\hat{\mathcal{N}}_{S'\rightarrow E}(\ketbra{\Phi}_{SS'})} = -\log \tr_{S'}\left( \left(\tr_S\left(\sqrt{\mathcal{N}_{S'}(\ketbra{\Phi}_{SS'})}\right)\right)^2\right).
\end{equation}

\subsection{Random operations and states, twirling, and design}
In this part, we introduce the concept of twirling and the design associated with random unitary operations. The twirling operation is defined over an ensemble of quantum unitary operations, $\mathfrak{S}$. We consider the following $t$-th order twirling operation over $\mathfrak{S}$:
\begin{equation}
\Phi^t_{\mathfrak{S}}(O) = \mathbb{E}_{U\sim\mathfrak{S}} U^{\otimes t}O U^{\dagger\otimes t}.
\end{equation}
where $U$ is drawn from the ensemble $\mathfrak{S}$; $O$ is the matrix being twirled, and its support may overpass $U^{\otimes t}$. The expectation is taken over the random unitary operations from $\mathfrak{S}$. When the ensemble $\mathfrak{S}$ corresponds to the Haar random unitaries $\mathfrak{U}_{H,n}$ on $n$ qubits, the first-order twirling operation over $\mathfrak{U}_{H,n}$ can be obtained from the Schur-Weyl duality \cite{fulton2004Representation}:
\begin{equation}\label{eq:firsttwirling}
\mathbb{E}_{U\sim \mathfrak{U}_{H,n}} U O U^{\dagger} = \tr_U(O)\frac{\id_{d}}{d},
\end{equation}
where $\tr_U$ denotes the trace over the subsystem being twirled, $d = \dim U$ is the dimension of Hilbert space that the unitary acts on, and $\id_{d}$ is the identity operator with dimension $d$. The second-order twirling operation over $\mathfrak{U}_{H,n}$ is
\begin{equation}\label{eq:secondtwirling}
\begin{split}
\mathbb{E}_{U\sim \mathfrak{U}_{H,n}} U^{\otimes 2} O U^{\dagger \otimes 2} = \tr_{U^{\otimes 2}}(O \frac{\id_{d^2}-d^{-1} S}{d^2-1})\id_{d^2}+\tr_{U^{\otimes 2}}(O \frac{S-d^{-1} \id_{d^2}}{d^2-1})S,
\end{split}
\end{equation}
with $\tr_{U^{\otimes 2}}$ denoting the trace over the subsystem of being twirled, $\id_{d^2}$ is the identity operator on the $d^2$-dimension subsystem, and $S$ denotes the SWAP operator on the $d^2$-dimension subsystem.

Though the Haar random unitary ensemble possesses good twirling properties, it requires exponential resources to implement in practice. An approximate version of the Haar random unitary ensemble is usually used for practical applications. An ensemble $\mathfrak{S}$ is called a unitary $t$-design if $\Phi^t_{\mathfrak{S}} = \Phi^t_{\mathfrak{U}_{H,n}}$. It was shown that the $n$-qubit Clifford group is an exact unitary 3-design~\cite{Webb2016Clifford3design,zhu2017threedesign}, which is automatically a unitary 1-design and 2-design. The 2-design ensemble have a good decoupling property as shown below.

\begin{lemma}[Decoupling theorem for unitary 2-design~\cite{Dupuis2014Decoupling,Liu2026AQEC}]\label{lemma:decouple2design} Suppose $U$ is drawn randomly from a unitary 2-design, then
\begin{equation}
\text{(Non-smooth version) } \mathbb{E}_{U}\Vert \mathcal{N}_{S\rightarrow E}(U_S\rho_{SR}U_{S}^{\dagger}) - \mathbb{E}_{U}\mathcal{N}_{S\rightarrow E}(U_S\rho_{S} U_S^{\dagger})\otimes \rho_R \Vert_1
\leq 2^{-\frac{1}{2}H_2(S|R)_{\rho_{SR}}-\frac{1}{2}H_2(S|E)_{\mathcal{N}_{S\rightarrow E}}},
\end{equation}
and
\begin{equation}
\text{(Smooth version) } \mathbb{E}_{U}\Vert \mathcal{N}_{S\rightarrow E}(U_S\rho_{SR}U_{S}^{\dagger}) - \mathbb{E}_{U}\mathcal{N}_{S\rightarrow E}(U_S\rho_{S} U_S^{\dagger})\otimes \rho_R \Vert_1
\leq 2^{-\frac{1}{2}H_2(S|R)_{\rho_{SR}}-\frac{1}{2}H^{\delta}_2(S|E)_{\mathcal{N}_{S\rightarrow E}}}+4\delta.
\end{equation}
\end{lemma}

Similar to the concept of Haar-random unitary operations and unitary $t$-design, one can introduce $n$-qubit Haar-random states, $\mathfrak{S}_{H,n} = \{U\ketbra{0^n}U^{\dagger}, U\in\mathfrak{U}_{H,n}\}$, and state $t$-design, $\mathfrak{S}$ such that $\mathbb{E}_{\psi\in \mathfrak{S}}\ketbra{\psi}^{\otimes t} = \mathbb{E}_{\psi\in \mathfrak{S}_{H,n}}\ketbra{\psi}^{\otimes t}$. For Haar-random states, we have Levy's lemma~\cite{ledoux2001concentration} stating that the continuous function of a Haar-random state is close to its average over the Haar measure.

\begin{lemma}[Levy's lemma]\label{Lemma:Levy}
Let $\psi$ be an $n$-qubit Haar-random state. For any function $f$ that is $\eta$-Lipschitz relative to the Euclidean norm, or $\abs{f(\psi)-f(\phi)}\leq \eta\Vert \ket{\psi}-\ket{\phi}\Vert_2$, the probability that $f(\psi)$ deviates from its expectation is bounded by:
\begin{equation}
    \mathrm{Pr}_{\psi} \left( \abs{f(\psi) - \mathbb{E}_{\psi}[f(\psi)]} \geq \epsilon \right) \leq 2 e^{- \frac{d \epsilon^2}{C \eta^2}},
\end{equation}
where $d = 2^n$ is the Hilbert space dimension and $C = 9\pi^3/4$ is a geometric constant for the state manifold.
\end{lemma}

In this paper, we need to consider a function $f(\psi_1, \psi_2)$ involving two independently Haar-random states. In this case, we need to generalize the original Levy's lemma to the product space. The result is presented below.

\begin{lemma}[Levy's Lemma for product Haar measure]\label{Lemma:LevyProduct}
Let $\psi_1,\psi_2$ be two $n$-qubit Haar-random states. For any function $f(\psi_1, \psi_2)$ that is $\eta_1$-Lipschitz in $\ket{\psi_1}$ and $\eta_2$-Lipschitz in $\ket{\psi_2}$, or $\abs{f(\psi_1, \psi_2)-f(\phi_1, \phi_2)}\leq \eta_1\Vert \ket{\psi_1}-\ket{\phi_1}\Vert_2 + \eta_2\Vert \ket{\psi_2}-\ket{\phi_2}\Vert_2$, the probability that $f(\psi_1, \psi_2)$ deviates from its expectation is bounded by
\begin{equation}
\mathrm{Pr}_{\psi_1, \psi_2} \left( |f(\psi_1, \psi_2) - \mathbb{E}_{\psi_1, \psi_2}[f(\psi_1, \psi_2)]| \geq \epsilon \right) \leq 2( e^{ -\frac{d \epsilon^2}{4C \eta_1^2} } + e^{-\frac{d \epsilon^2}{4C \eta_2^2} } ).
\end{equation}
where $d = 2^n$ is the Hilbert space dimension and $C = 9\pi^3/4$ is a geometric constant for the state manifold.
\end{lemma}
\begin{proof}[Proof of Lemma~\ref{Lemma:LevyProduct}]
Define the conditional expectation $g(\psi_2) = \mathbb{E}_{\psi_1}[f(\psi_1, \psi_2)]$, where the expectation is taken over $\psi_1$ with $\psi_2$ fixed. By the triangle inequality,
\begin{equation}
|f(\psi_1, \psi_2) - \mathbb{E}[f]| \leq |f(\psi_1, \psi_2) - g(\psi_2)| + |g(\psi_2) - \mathbb{E}[f]|,
\end{equation}
where $\mathbb{E}[f] = \mathbb{E}_{\psi_1, \psi_2}[f(\psi_1, \psi_2)]$. Therefore, the event $|f(\psi_1, \psi_2) - \mathbb{E}[f]| \geq \epsilon$ implies that at least one of $|f(\psi_1, \psi_2) - g(\psi_2)| \geq \epsilon/2$ or $|g(\psi_2) - \mathbb{E}[f]| \geq \epsilon/2$ occurs. By the union bound,
\begin{equation}
\Pr\left( |f(\psi_1, \psi_2) - \mathbb{E}[f]| \geq \epsilon \right) \leq \Pr\left( |f(\psi_1, \psi_2) - g(\psi_2)| \geq \epsilon/2 \right) + \Pr\left( |g(\psi_2) - \mathbb{E}[f]| \geq \epsilon/2 \right). \label{eq:union}
\end{equation}
We now bound each term separately. For any fixed $\psi_2$, consider the function $h_{\psi_2}(\psi_1) = f(\psi_1, \psi_2)$. From the Lipschitz condition in Lemma~\ref{Lemma:LevyProduct}, for any $\psi_1, \phi_1$ we have
\begin{equation}
|h_{\psi_2}(\psi_1) - h_{\psi_2}(\phi_1)| \leq \eta_1 \Vert \ket{\psi_1} - \ket{\phi_1}\Vert_2.
\end{equation}
Thus $h_{\psi_2}$ is $\eta_1$-Lipschitz in $\psi_1$. Applying Lemma~\ref{Lemma:Levy} to $h_{\psi_2}$ gives, for any fixed $\psi_2$,
\begin{equation}
\Pr_{\psi_1}\left( |f(\psi_1, \psi_2) - g(\psi_2)| \geq \epsilon/2 \right) \leq 2 \exp\left( -\frac{d (\epsilon/2)^2}{C \eta_1^2} \right) = 2 \exp\left( -\frac{d \epsilon^2}{4C \eta_1^2} \right).
\end{equation}
This bound holds for every $\psi_2$, so it remains valid when taking the probability over both $\psi_1$ and $\psi_2$. Hence
\begin{equation}\label{eq:first}
\Pr\left( |f(\psi_1, \psi_2) - g(\psi_2)| \geq \epsilon/2 \right)\leq 2 \exp\left( -\frac{d \epsilon^2}{4C \eta_1^2} \right).
\end{equation}
Meanwhile, $g(\psi_2)$ is $\eta_2$-Lipschitz. For any $\psi_2, \phi_2$,
\begin{equation}\label{eq:g_lipschitz}
\begin{split}
|g(\psi_2) - g(\phi_2)| &= \left| \mathbb{E}_{\psi_1}[f(\psi_1, \psi_2)] - \mathbb{E}_{\psi_1}[f(\psi_1, \phi_2)] \right| \\
&\leq \mathbb{E}_{\psi_1} \left[ |f(\psi_1, \psi_2) - f(\psi_1, \phi_2)| \right]   \\
&\leq \mathbb{E}_{\psi_1} \left[ \eta_2\Vert \ket{\psi_2}-\ket{\phi_2}\Vert_2 \right] = \eta_2\Vert \ket{\psi_2}-\ket{\phi_2}\Vert_2,
\end{split}
\end{equation}
Applying Lemma~\ref{Lemma:Levy} to $g(\psi_2)$ as a function of $\psi_2$ yields
\begin{equation}
\Pr_{\psi_2}\left( |g(\psi_2) - \mathbb{E}[f]| \geq \epsilon/2 \right) \leq 2 \exp\left( -\frac{d \epsilon^2}{4C \eta_2^2} \right). \label{eq:second}
\end{equation}
Thus,
\begin{equation}\label{eq:combined}
\Pr\left( |f(\psi_1, \psi_2) - \mathbb{E}[f]| \geq \epsilon \right) \leq 2 \left[ \exp\left( -\frac{d \epsilon^2}{4C \eta_1^2} \right) + \exp\left( -\frac{d \epsilon^2}{4C \eta_2^2} \right) \right].
\end{equation}
\end{proof}

\section{Bounds of average pairwise trace distance for noisy state ensembles}\label{appendsc:pairwise}
Here, we prove tight upper and lower bounds for the average pairwise distance of a state ensemble. After that, we evaluate the value of the average pairwise distance of a 2-design ensemble under noise.
\subsection{Bounds of average pairwise distance}\label{appendssc:pairwisebound}
Below, we show that the decoupling approach can provide both the upper bounds and lower bounds for the average pairwise distance, $\mathbb{E}_{\psi_1\in \mathfrak{S}_1, \psi_2\in \mathfrak{S}_2}\Vert \mathcal{N}(\psi_1) - \mathcal{N}(\psi_2)\Vert_1$. We first present a more general theorem than that in the main text as follows. Then we give the proof.

\begin{theorem}\label{theo:noisyhypo}
Suppose state set $\mathfrak{S}_i$ can be generated by unitary ensemble $\mathfrak{U}_i$, $\mathfrak{S}_i = \{U\ketbra{0^n}U^{\dagger}, U\in \mathfrak{U}_i\}$, then
\begin{equation}
\begin{split}
\mathbb{E}_{\psi_1\sim \mathfrak{S}_1, \psi_2\sim \mathfrak{S}_2}\Vert \mathcal{N}(\psi_1) - \mathcal{N}(\psi_2)\Vert_1 \leq \epsilon_{\mathcal{N}, \mathfrak{U}_1, \ketbra{0^n}} + \epsilon_{\mathcal{N}, \mathfrak{U}_2, \ketbra{0^n}}.
\end{split}
\end{equation}
Moreover, suppose there exists a unitary ensemble $\mathfrak{U}=\{p(U),U\}$ such that $\mathfrak{S}_i$ is invariant under any unitary operation $U$ from $\mathfrak{U}$, $\mathfrak{S}_i=\{p_i, \psi_i\}=\{p_i, U\psi_i U^{\dagger}\}$, then
\begin{equation}
\mathbb{E}_{\psi_1,\psi_2}\Vert\mathcal{N}(\ketbra{\psi_1})-\mathcal{N}(\ketbra{\psi_2})\Vert_1 \geq \mathbb{E}_{\psi_1,\psi_2}\Vert \ketbra{\psi_1} - \ketbra{\psi_2} \Vert_1 - 8\sqrt{\max_W \epsilon_{\hat{\mathcal{N}}, \mathfrak{U}_W, \ketbra{\Phi}\otimes \ketbra{0^{n-1}}}}.
\end{equation}
Here, we denote $\mathfrak{U}_W = \{UW, U\in \mathfrak{U}\}$ while $W$ is maximized over all unitary operations. Also,
\begin{equation}
\epsilon_{\mathcal{N}, \mathfrak{U}, \rho} \equiv \mathbb{E}_{U\in \mathfrak{U}}\Vert \mathcal{N}_{S\rightarrow E}(U_S\rho_{SR} U_S^{\dagger}) - \mathbb{E}_{U\in \mathfrak{U}}\mathcal{N}_{S\rightarrow E}(U_S\rho_{S} U_S^{\dagger})\otimes \rho_R\Vert_1
\end{equation}
is the decoupling error corresponding to channel $\mathcal{N}$, initial state $\rho$ and unitary ensemble $\mathfrak{U}$.
\end{theorem}

\subsubsection{Upper bound}
Employing the triangle inequality, one can obtain that
\begin{equation}
\begin{split}
\mathbb{E}_{\psi_1\in \mathfrak{S}_1, \psi_2\in \mathfrak{S}_2}\Vert \mathcal{N}(\psi_1) - \mathcal{N}(\psi_2)\Vert_1 &\leq \mathbb{E}_{\psi_1,\psi_2}\Vert\mathcal{N}(\psi_1)-\mathcal{N}(\frac{\id}{2^n})\Vert_1  + \mathbb{E}_{\psi_1,\psi_2}\Vert\mathcal{N}(\psi_2)-\mathcal{N}(\frac{\id}{2^n})\Vert_1\\
&= \mathbb{E}_{\psi_1}\Vert\mathcal{N}(\psi_1)-\mathcal{N}(\frac{\id}{2^n})\Vert_1  + \mathbb{E}_{\psi_2}\Vert\mathcal{N}(\psi_2)-\mathcal{N}(\frac{\id}{2^n})\Vert_1.
\end{split}
\end{equation}
This paper considers that the sets of quantum states satisfy a 1-design property. That is $\mathbb{E}_{\psi\in \mathfrak{S}}\ketbra{\psi} = \frac{\id}{2^n}$. Thus, $\forall i\in \{1,2\}$, $\mathbb{E}_{\psi_i}\Vert\mathcal{N}(\psi_i)-\mathcal{N}(\frac{\id}{2^n})\Vert_1$ is the variance of $\mathcal{N}(\psi_i)$ deviating from its average value, or the decoupling error of $\mathfrak{S}_i$ on channel $\mathcal{N}$. Bounding this decoupling error suffices to provide upper bounds of the original average pairwise distance. Particularly, when $\mathfrak{S}$ is generated by a unitary ensemble $\mathfrak{U}$, $\mathfrak{S} = \{U\ketbra{0^n} U^{\dagger}, U\in \mathfrak{U}\}$, we have that
\begin{equation}
\begin{split}
\mathbb{E}_{\psi}\Vert\mathcal{N}(\psi)-\mathcal{N}(\frac{\id}{2^n})\Vert_1 \leq \mathbb{E}_U\Vert\mathcal{N}(U\ketbra{0^n} U^{\dagger})-\mathcal{N}(\frac{\id}{2^n})\Vert_1 \equiv \epsilon_{\mathcal{N}, \mathfrak{U}, \ketbra{0^n}},
\end{split}
\end{equation}
where $\epsilon_{\mathcal{N}, \mathfrak{U}}$ is the decoupling error corresponding to channel $\mathcal{N}$, initial state $\ketbra{0^n}$ and unitary ensemble $\mathfrak{U}$.

\subsubsection{Lower bound}
We begin by changing the way of expectation. Note that $\mathfrak{S}_i=\{p_i, \psi_i\}=\{p_i, U\psi_i U^{\dagger}\}$. The average distance between two sets can be evaluated by
\begin{equation}
\mathbb{E}_{\psi_1,\psi_2}\Vert\mathcal{N}(\ketbra{\psi_1})-\mathcal{N}(\ketbra{\psi_2})\Vert_1 = \mathbb{E}_{\psi_1,\psi_2}\mathbb{E}_{U}
\Vert \mathcal{N}\left(U \ketbra{\psi_1} U^\dagger\right) - \mathcal{N}\left(U \ketbra{\psi_2} U^\dagger\right) \Vert_1.
\end{equation}
Here, gate $U$ is randomly sampled from $\mathfrak{U}$, and state $\psi_i$ is randomly sampled from $\mathfrak{S}_i$.

For any pair of states $\psi_1$ and $\psi_2$, we will show that
\begin{equation}
\mathbb{E}_{U}
\Vert \mathcal{N}\left(U \ketbra{\psi_1} U^\dagger\right) - \mathcal{N}\left(U \ketbra{\psi_2} U^\dagger\right) \Vert_1 \geq \Vert \ketbra{\psi_1} - \ketbra{\psi_2} \Vert_1 - 8\sqrt{\max_W \epsilon_{\hat{\mathcal{N}}, \mathfrak{U}_W, \ketbra{\Phi}\otimes \ketbra{0^{n-1}}}},
\end{equation}
where $\epsilon_{\hat{\mathcal{N}}, \mathfrak{U}_W, \ketbra{\Phi}\otimes \ketbra{0^{n-1}}}$ is the decoupling error corresponding to channel $\hat{\mathcal{N}}$, an initial state $\ketbra{\Phi}\otimes \ketbra{0^{n-1}}$ and unitary ensemble $\mathfrak{U}_W$ labled with a unitary $W$. The concrete expression is explained below.

Particularly, for any CPTP map $\mathcal{D}$ and unitary $V$, we have that
\begin{equation}
\begin{split}
&\Vert \mathcal{N}\left(U \ketbra{\psi_1} U^\dagger\right) - \mathcal{N}\left(U \ketbra{\psi_2} U^\dagger\right) \Vert_1\geq\\
&\Vert \mathcal{D}\circ\mathcal{N}\left(U \ketbra{\psi_1} U^\dagger\right) - \mathcal{D}\circ\mathcal{N}\left(U \ketbra{\psi_2} U^\dagger\right) \Vert_1\geq\\
&\Vert V\ketbra{\psi_1}V^{\dagger} - V\ketbra{\psi_2}V^{\dagger} \Vert_1 - \Vert \mathcal{D}\circ\mathcal{N}\left(U \ketbra{\psi_1} U^\dagger\right) - V\ketbra{\psi_1}V^{\dagger} \Vert_1 - \Vert\mathcal{D}\circ\mathcal{N}\left(U \ketbra{\psi_2} U^\dagger\right) - V\ketbra{\psi_2}V^{\dagger}\Vert_1=\\
&\Vert \ketbra{\psi_1} - \ketbra{\psi_2} \Vert_1 - \Vert \mathcal{D}\circ\mathcal{N}\left(U \ketbra{\psi_1} U^\dagger\right) - V\ketbra{\psi_1}V^{\dagger} \Vert_1 - \Vert\mathcal{D}\circ\mathcal{N}\left(U \ketbra{\psi_2} U^\dagger\right) - V\ketbra{\psi_2}V^{\dagger}\Vert_1.
\end{split}
\end{equation}

The remaining step is to upper bound
\begin{equation}
\mathbb{E}_{U} \min_{\mathcal{D}, V}\left(\Vert \mathcal{D}\circ\mathcal{N}\left(U \ketbra{\psi_1} U^\dagger\right) - V\ketbra{\psi_1}V^{\dagger} \Vert_1 + \Vert\mathcal{D}\circ\mathcal{N}\left(U \ketbra{\psi_2} U^\dagger\right) - V\ketbra{\psi_2}V^{\dagger}\Vert_1\right).
\end{equation}

To give an effective upper bound, we employ the following idea: we view the two given states $\ket{\psi_1}$ and $\ket{\psi_2}$ as elements of a two-dimensional logical subspace. Note that these two states are not necessarily orthogonal to each other.

To see that, we first fix a logical subspace
\begin{equation}
\mathcal{C} = \mathrm{span}\{\ket{0_L}\otimes\ket{0^{n-1}},\ \ket{1_L}\otimes\ket{0^{n-1}}\}.
\end{equation}
We show that there exists a unitary $W$ such that $W\ket{\psi_1}$ and $W\ket{\psi_2}$ both lie in $\mathcal{C}$.
This can be proved by a direct construction.

If $\ket{\psi_2}$ is proportional to $\ket{\psi_1}$, then $\mathrm{span}\{\ket{\psi_1},\ket{\psi_2}\}$ is one-dimensional.
We can choose any unitary $W$ satisfying $\ket{0_L}\otimes\ket{0^{n-1}}=W\ket{\psi_1}$,
and extend $W$ arbitrarily on the orthogonal complement.
Then $W\ket{\psi_1}, W\ket{\psi_2}\in\mathcal{C}$.

If $\ket{\psi_1}$ and $\ket{\psi_2}$ are linearly independent, apply Gram–Schmidt to obtain an orthonormal pair $\{\ket{u_0},\ket{u_1}\}$ such that
\begin{equation}
\mathrm{span}\{\ket{u_0},\ket{u_1}\} = \mathrm{span}\{\ket{\psi_1},\ket{\psi_2}\}.
\end{equation}
Explicitly,
\begin{equation}
\ket{u_0} = \ket{\psi_1}, \quad
\ket{u_1} = \frac{\ket{\psi_2} - \langle u_0|\psi_2\rangle \ket{u_0}}{\Vert\ket{\psi_2} - \langle u_0|\psi_2\rangle \ket{u_0}\Vert_2}.
\end{equation}

Let $\{\ket{e_0},\ket{e_1}\}$ be an orthonormal basis of $\mathcal{C}$, where
\begin{equation}
\ket{e_0} = \ket{0_L}\otimes\ket{0^{n-1}}, \quad
\ket{e_1} = \ket{1_L}\otimes\ket{0^{n-1}}.
\end{equation}
Extend $\{\ket{e_0},\ket{e_1}\}$ and $\{\ket{u_0},\ket{u_1}\}$ to full orthonormal bases of the Hilbert space.

Define $W$ by
\begin{equation}
W\ket{u_j} = \ket{e_j}, \quad j=0,\dots,2^n-1.
\end{equation}
Then $W$ is unitary, and $W\ket{\psi_1} = W\ket{u_0} = \ket{e_0}\in\mathcal{C}$.
Since $\ket{\psi_2} = a_{\psi}\ket{u_0} + b_{\psi}\ket{u_1}$ for some $a_{\psi},b_{\psi}$,
we have
\begin{equation}
W\ket{\psi_2} = (a_{\psi}\ket{0_L}+b_{\psi}\ket{1_L})\otimes\ket{0^{n-1}} \in \mathcal{C}.
\end{equation}
Therefore, there always exists a unitary $W$ such that
$W\ket{\psi_1}$ and $W\ket{\psi_2}$ both belong to $\mathcal{C}$. For simplicity, in the following, we use $\ket{\psi_L}$ to represent $a_{\psi}\ket{0_L}+b_{\psi}\ket{1_L}$.

By choosing $V$ as $W$, we get that
\begin{equation}
\begin{split}
&\min_{V}\left(\Vert \mathcal{D}\circ\mathcal{N}\left(U \ketbra{\psi_1} U^\dagger\right) - V\ketbra{\psi_1}V^{\dagger} \Vert_1 + \Vert\mathcal{D}\circ\mathcal{N}\left(U \ketbra{\psi_2} U^\dagger\right) - V\ketbra{\psi_2}V^{\dagger}\Vert_1\right)\leq\\
&\Vert \mathcal{D}\circ\mathcal{N}\left(UW^{\dagger} \ketbra{0_L}\otimes \ketbra{0^{n-1}} WU^\dagger\right) - \ketbra{0_L}\otimes \ketbra{0^{n-1}} \Vert_1 +\\
&\Vert \mathcal{D}\circ\mathcal{N}\left(UW^{\dagger} \ketbra{\psi_L}\otimes \ketbra{0^{n-1}} WU^\dagger\right) - \ketbra{\psi_L}\otimes \ketbra{0^{n-1}} \Vert_1.
\end{split}
\end{equation}

Denote $\mathfrak{U}_W = \{UW, U\in\mathfrak{U}\}$, our target turns to upper bound
\begin{equation}
\begin{split}
\mathbb{E}_{U\sim \mathfrak{U}_W} \min_{\mathcal{D}}&(\Vert \mathcal{D}\circ\mathcal{N}\left(U \ketbra{0_L}\otimes \ketbra{0^{n-1}} U^\dagger\right) - \ketbra{0_L}\otimes \ketbra{0^{n-1}} \Vert_1 +\\
&\Vert \mathcal{D}\circ\mathcal{N}\left(U \ketbra{\psi_L}\otimes \ketbra{0^{n-1}} U^\dagger\right) - \ketbra{\psi_L}\otimes \ketbra{0^{n-1}} \Vert_1).
\end{split}
\end{equation}
By viewing the CPTP map $\mathcal{D}$ as the decoding map associated with encoding unitary $U$, this is exactly the decoding error against noise $\mathcal{N}$ under the random encoding ensemble $\mathfrak{U}_W$ on a two-dimensional space. Clearly, the target value is upper-bounded by
\begin{equation}
2\mathbb{E}_{U\sim \mathfrak{U}_W} \min_{\mathcal{D}}\max_{\psi}(\Vert \mathcal{D}\circ\mathcal{N}\left(U \ketbra{\psi_L}\otimes \ketbra{0^{n-1}} U^\dagger\right) - \ketbra{\psi_L}\otimes \ketbra{0^{n-1}} \Vert_1).
\end{equation}
The term for expectation is the worst-case error over all input states within the logical system:
\begin{equation}
\begin{split}
&\min_{\mathcal{D}}\max_{\psi}(\Vert \mathcal{D}\circ\mathcal{N}\left(U \ketbra{\psi_L}\otimes \ketbra{0^{n-1}} U^\dagger\right) - \ketbra{\psi_L}\otimes \ketbra{0^{n-1}} \Vert_1)\\
\leq&\min_{\mathcal{D}_L}\max_{\psi}(\Vert \mathcal{D}_L\circ\mathcal{N}\left(U \ketbra{\psi_L}\otimes \ketbra{0^{n-1}} U^\dagger\right)\otimes \ketbra{0^{n-1}} - \ketbra{\psi_L}\otimes \ketbra{0^{n-1}} \Vert_1)\\
=&\min_{\mathcal{D}_L}\max_{\psi}(\Vert \mathcal{D}_L\circ\mathcal{N}\left(U \ketbra{\psi_L}\otimes \ketbra{0^{n-1}} U^\dagger\right) - \ketbra{\psi_L} \Vert_1)\\
\leq&
\min_{\mathcal{D}_L} \Vert \mathcal{D}_L\circ\mathcal{N}\circ\mathcal{E} - \mathcal{I}\Vert_{\diamond}.
\end{split}
\end{equation}
Here, $\mathcal{D}_L$ is a decoder that only output one qubit as the logical qubit, $\mathcal{E}(\rho) = U\rho\otimes \ketbra{0^{n-1}} U^{\dagger}$ is the encoding map, and $\mathcal{I}$ is the identity map on the density matrix of a two-dimensional Hilbert space. The second line is because decoders outputing a single qubit is only a subset of all possible decoders.

There is a relationship between diamond distance and entanglement fidelity~\cite{Wallman2014benchmarking}:
\begin{equation}
\Vert \mathcal{T} - \mathcal{I} \Vert_{\diamond}\leq 2d\sqrt{1-F^2(\mathcal{T})}
\end{equation}
where $F^2(\mathcal{T}) = \tr(\ketbra{\Phi}\mathcal{T}(\ketbra{\Phi}))$ is the entanglement fidelity of $\mathcal{T}$, and $\ket{\Phi}$ is the maximally entangled state. Since here $\mathcal{D}\circ\mathcal{N}\circ\mathcal{E}$ act on a two-dimensional system, we have that $d = 2$. Thus,
\begin{equation}
\min_{\mathcal{D}} \Vert \mathcal{D}\circ\mathcal{N}\circ\mathcal{E} - \mathcal{I}\Vert_{\diamond}\leq 4 \min_{\mathcal{D}} \sqrt{1-F^2(\mathcal{D}\circ\mathcal{N}\circ\mathcal{E})} = 4 \min_{\mathcal{D}} P(\ketbra{\Phi}, \mathcal{D}\circ\mathcal{N}\circ\mathcal{E}(\ketbra{\Phi})) = 4\epsilon_{\mathrm{Choi}},
\end{equation}
where $P(\rho, \sigma)$ is the purified distance, and by definition, the above quantity is equivalent to the Choi error $\epsilon_{\mathrm{Choi}}$ of encoding one logical qubit with $\mathcal{E}$ against noise $\mathcal{N}$. By the complementary channel approach~\cite{Beny2010AQEC}, this Choi error can be characterized by the complementary channel of the noise:
\begin{equation}
\begin{split}
\epsilon_{\mathrm{Choi}} &= \min_{\mathcal{D}} P(\ketbra{\Phi_{SR}}, (\mathcal{D}\circ\mathcal{N}\circ\mathcal{E})_S(\ketbra{\Phi_{SR}}))\\
&= \min_{\mathcal{\xi}} P(\xi_E\otimes \frac{\id_R}{2}, \hat{\mathcal{N}}_{S\rightarrow E}\circ\mathcal{E}_S(\ketbra{\Phi_{SR}}))\\
&\leq \min_{\mathcal{\xi}}\sqrt{\Vert \hat{\mathcal{N}}_{S\rightarrow E}\circ\mathcal{E}_S(\ketbra{\Phi_{SR}}) - \xi_E\otimes \frac{\id_R}{2}\Vert_1}.
\end{split}
\end{equation}
Thus, our target quantity can be upper-bounded by
\begin{equation}
\begin{split}
&8\mathbb{E}_{U\sim \mathfrak{U}_W} \min_{\mathcal{\xi}}\sqrt{\Vert \hat{\mathcal{N}}_{S\rightarrow E}\circ\mathcal{E}_S(\ketbra{\Phi_{SR}}) - \xi_E\otimes \frac{\id_R}{2}\Vert_1}\\
\leq &8 \sqrt{\mathbb{E}_{U\sim \mathfrak{U}_W}\min_{\mathcal{\xi}}\Vert \hat{\mathcal{N}}_{S\rightarrow E}\circ\mathcal{E}_S(\ketbra{\Phi_{SR}}) - \xi_E\otimes \frac{\id_R}{2}\Vert_1}\\
\leq& 8\sqrt{ \mathbb{E}_{U\sim \mathfrak{U}_W}\Vert \hat{\mathcal{N}}_{S\rightarrow E}\circ\mathcal{E}_S(\ketbra{\Phi_{SR}}) - \hat{\mathcal{N}}_{S\rightarrow E}\circ\mathcal{E}_S(\frac{\id_S}{2})\otimes \frac{\id_R}{2}\Vert_1 }\\
\leq& 8\sqrt{\max_W \mathbb{E}_{U\sim \mathfrak{U}_W}\Vert \hat{\mathcal{N}}_{S\rightarrow E}\circ\mathcal{E}_S(\ketbra{\Phi_{SR}}) - \hat{\mathcal{N}}_{S\rightarrow E}\circ\mathcal{E}_S(\frac{\id_S}{2})\otimes \frac{\id_R}{2}\Vert_1}\\
=& 8\sqrt{\max_W \epsilon_{\hat{\mathcal{N}}, \mathfrak{U}_W, \ketbra{\Phi}\otimes \ketbra{0^{n-1}}}}.
\end{split}
\end{equation}
Here,
\begin{equation}
\epsilon_{\hat{\mathcal{N}}, \mathfrak{U}_W, \ketbra{\Phi}\otimes \ketbra{0^{n-1}}} \equiv \mathbb{E}_{U\sim \mathfrak{U}_W}\Vert \hat{\mathcal{N}}_{S\rightarrow E}\circ\mathcal{E}_S(\ketbra{\Phi_{SR}}) - \hat{\mathcal{N}}_{S\rightarrow E}\circ\mathcal{E}_S(\frac{\id_S}{2})\otimes \frac{\id_R}{2}\Vert_1
\end{equation}
is the decoupling error corresponding to channel $\hat{\mathcal{N}}$, initial state $\ketbra{\Phi}\otimes \ketbra{0^{n-1}}$ and unitary ensemble $\mathfrak{U}_W$. Note that the initial state only contains one EPR pair.

Thus, the final lower bound is given by
\begin{equation}
\mathbb{E}_{\psi_1,\psi_2}\Vert\mathcal{N}(\ketbra{\psi_1})-\mathcal{N}(\ketbra{\psi_2})\Vert_1 \geq \mathbb{E}_{\psi_1,\psi_2}\Vert \ketbra{\psi_1} - \ketbra{\psi_2} \Vert_1 - 8\sqrt{\max_W \epsilon_{\hat{\mathcal{N}}, \mathfrak{U}_W, \ketbra{\Phi}\otimes \ketbra{0^{n-1}}}}.
\end{equation}

\subsection{Average pairwise distance of two-design ensembles}
Here, we prove Theorem~\ref{thm:2design} where the two sets are invariant under the action of the same unitary 2-design ensemble. That is, $\mathfrak{S}_i=\{p_i, \psi_i\}=\{p(U), U\ketbra{0^n} U^{\dagger}\}$ where $\mathfrak{U} = \{p(U), U\}$ is a unitary 2-design. Obviously, in this case, $\mathfrak{S}_i$ is a state 2-design. Also, note that $\mathfrak{U}_W$ is always a unitary 2-design for arbitrary unitary $W$ provided that $\mathfrak{U}$ is a unitary 2-design. By applying the decoupling theorem of unitary 2-design, or Lemma~\ref{lemma:decouple2design}, to Theorem~\ref{theo:noisyhypo}, we get that
\begin{equation}\label{eq:2designupper}
\mathbb{E}_{\psi_1\in \mathfrak{S}_1, \psi_2\in \mathfrak{S}_2}\Vert \mathcal{N}(\psi_1) - \mathcal{N}(\psi_2)\Vert_1/2 \leq \epsilon_{\mathcal{N}, \mathfrak{U}, \ketbra{0^n}} \leq \begin{cases}
2^{-\frac{1}{2}H_2(S|S')_{\mathcal{N}}}, & \text{non-smooth}, \\
2^{-\frac{1}{2}H^{\delta}_2(S|S')_{\mathcal{N}}}+4\delta,  & \text{smooth}.
\end{cases}.
\end{equation}
\begin{equation}\label{eq:2designlower}
\begin{split}
\mathbb{E}_{\psi_1\in \mathfrak{S}_1, \psi_2\in \mathfrak{S}_2}\Vert \mathcal{N}(\psi_1) - \mathcal{N}(\psi_2)\Vert_1/2 &\geq \mathbb{E}_{\psi_1\in \mathfrak{S}_1, \psi_2\in \mathfrak{S}_2}\Vert \psi_1 - \psi_2\Vert_1 / 2 - 4 \sqrt{\max_W \epsilon_{\hat{\mathcal{N}}, \mathfrak{U}_W, \ketbra{\Phi}\otimes \ketbra{0^{n-1}}}}\\
&= \mathbb{E}_{\psi_1\in \mathfrak{S}_1, \psi_2\in \mathfrak{S}_2} \sqrt{1-\abs{\braket{\psi_1}{\psi_2}}^2} - 4 \sqrt{\max_W \epsilon_{\hat{\mathcal{N}}, \mathfrak{U}_W, \ketbra{\Phi}\otimes \ketbra{0^{n-1}}}}\\
&\geq \mathbb{E}_{\psi_1\in \mathfrak{S}_1, \psi_2\in \mathfrak{S}_2} (1-\abs{\braket{\psi_1}{\psi_2}}^2) - 4 \sqrt{\max_W \epsilon_{\hat{\mathcal{N}}, \mathfrak{U}_W, \ketbra{\Phi}\otimes \ketbra{0^{n-1}}}}\\
&\geq (1-\tr((\frac{\id}{2^n})^2)) - 4 \sqrt{\max_W \epsilon_{\hat{\mathcal{N}}, \mathfrak{U}_W, \ketbra{\Phi}\otimes \ketbra{0^{n-1}}}}\\
&\geq 1 - 2^{-n} - 4\times \begin{cases}
\sqrt{2^{\frac{1}{2}-\frac{1}{2}H_2(S|E)_{\hat{\mathcal{N}}}}}, & \text{non-smooth}, \\
\sqrt{2^{\frac{1}{2}-\frac{1}{2}H^{\delta}_2(S|E)_{\hat{\mathcal{N}}}}+4\delta},  & \text{smooth}.
\end{cases}.
\end{split}
\end{equation}
For \eqref{eq:2designlower}, the second line utilizes $\Vert \psi_1 - \psi_2 \Vert_1 = 2\sqrt{1- \abs{\braket{\psi_1}{\psi_2}}^2}$ for pure states. The fourth line utilizes the condition that $\mathfrak{S}_i$ is a state 2-design. The last lines of \eqref{eq:2designupper} and \eqref{eq:2designlower} utilize the results of the decoupling theorem. Here, $H_2(S|R)_{\ketbra{0^n}}=H_2(\ketbra{0^n})=0$ and $H_2(S|R)_{\ketbra{\Phi}_{sr}\otimes \ketbra{0^{n-1}}_{SR\backslash sr}}=H_2(s|r)_{\ketbra{\Phi}_{sr}}=-1$. Note that the upper bound \eqref{eq:2designupper} can also be derived as long as $\mathfrak{S}_i$ is a state 2-design without being generated by a unitary 2-design.

Based on Chebyshev's inequality $\Pr(\abs{X-\mathbb{E}[X]}\geq \epsilon)\leq \frac{\mathbb{V}[X]}{\epsilon^2}$, one can also obtain the high-probability version of the above result. Note that if $0\leq X\leq 1$, $\mathbb{V}[X] = \mathbb{E}[X^2]-(\mathbb{E}[X])^2\leq \mathbb{E}[X](1-\mathbb{E}[X])\leq \min(\mathbb{E}[X], 1-\mathbb{E}[X])$.

\begin{corollary}\label{coro:2designhighprob}
Suppose $\mathfrak{S}$ is generated by a unitary $2$-design ensemble, then
\begin{align}
&\Pr_{\psi_1,\psi_2\sim \mathfrak{S}}(\Vert\mathcal{N}(\psi_1)-\mathcal{N}(\psi_2)\Vert_1/2\geq \epsilon) \leq \begin{cases}
\frac{2^{-\frac{1}{2}H_2(\mathcal{N})}}{(\epsilon-2^{-\frac{1}{2}H_2(\mathcal{N})})^2}, \\
\frac{2^{-\frac{1}{2}H^{\delta}_2(\mathcal{N})}+4\delta}{(\epsilon-2^{-\frac{1}{2}H^{\delta}_2(\mathcal{N})}-4\delta)^2}.
\end{cases};\\
&\Pr_{\psi_1,\psi_2\sim \mathfrak{S}}(\Vert\mathcal{N}(\psi_1)-\mathcal{N}(\psi_2)\Vert_1 / 2\leq 1-\epsilon) \leq \begin{cases}
\frac{2^{-n}+2^{\frac{9}{4}-\frac{1}{4}H_2(\hat{\mathcal{N}})}}{(\epsilon-2^{-n}-2^{\frac{9}{4}-\frac{1}{4}H_2(\hat{\mathcal{N}})})^2}, \\
\frac{2^{-n}+4\sqrt{2^{\frac{1}{2}-\frac{1}{2}H_2^{\delta}(\hat{\mathcal{N}})} + 4\delta}}{(\epsilon-2^{-n}-4\sqrt{2^{\frac{1}{2}-\frac{1}{2}H_2^{\delta}(\hat{\mathcal{N}})} + 4\delta})^2}.
\end{cases}.
\end{align}
\end{corollary}

\section{Distinguishability of measured state ensembles and noise tolerance of information processing}
In this part, we first show that the average pairwise distance of a (measured) state ensemble provides both upper and lower bounds for the mutual information if we use this state ensemble as the encoding codebook or POVM. Combining with the result of the average pairwise distance, we show that the use of a 2-design for information processing is not robust to noise. There is no noise threshold for local purity-shrinking noise. At the end of this section, we present another proof that further strengthens the result. Notice that in the main text, the bound on the mutual information in Theorem~\ref{thm:pairwisedistance2mutualinformation} is a single-copy version. Here, we prove a more generic result of the $N$-copy version.

\subsection{Relation between mutual information and average pairwise distance}\label{appendssc:pairwisemulinf}
We first discuss the case where Alice uses a noisy state ensemble $\mathfrak{E} = \{p_x, \rho_x\}$ to encode the classical message $X$, sends one of the states to Bob, and Bob measures $\rho_X$ with a POVM $\mathfrak{F} = \{F_y\}$ to get a measurement outcome. One can view the sending state $\rho_x = \mathcal{N}(\psi_x)$ as the original noiseless state $\psi_x$ passing through the noise channel $\mathcal{N}$. They repeated this for $N$ rounds, and we denote the outcomes of Bob as $Y = Y^{(N)} = (Y_1, Y_2, \cdots, Y_N)$. The following derivations give the bounds of mutual information $I(X:Y)$ with the average pairwise distance of $\mathfrak{E}^F = \{p_x, \mathcal{F}(\rho_x)\}$, where Bob performs the measurement channel $\mathcal{F}$. One can optimize the measurement channel $\mathcal{F}$ to maximize the pairwise distance of $\mathfrak{E}^{\mathcal{F}}$. This optimized result provides the bounds of the accessible information of $\mathfrak{E}$. We later discuss on using the average pairwise distance of POVM $\mathfrak{F}$ to provide bounds by using the duality between states and POVM.
\subsubsection{Upper bound}
By the Holevo bound:
\begin{equation}
I(X: Y^{(N)}) \leq \chi(X^{\otimes N}) = S\left(\sum_{x} p_{x} \rho_{x}^{\otimes N}\right) - \sum_{x} p_{x} S\left(\rho_{x}^{\otimes N}\right),
\end{equation}
and for any $d$-dimensional quantum state $\rho,\sigma$, we have~\cite{audenaert2006sharp}
\begin{equation}
\abs{S(\rho)-S(\sigma)} \leq \frac{1}{2}\log(d-1)\Vert\rho-\sigma\Vert_1 + h_2\left(\frac{1}{2}\Vert\rho-\sigma\Vert_1\right),
\end{equation}
where $h_2$ is the binary entropy. Thus,
\begin{equation}
\begin{split}
I(X: Y^{(N)}) &\leq \frac{1}{2} \log (d-1) \sum_{x^{\prime}}p_{x^{\prime}}\Vert \sum_{x}p_x\rho_{x}^{\otimes N} - \rho_{x^{\prime}}^{\otimes N}\Vert_1 + \sum_{x^{\prime}}p_{x^{\prime}}h_2\left(\frac{1}{2}\Vert \sum_{x}p_x\rho_{x}^{\otimes N} - \rho_{x^{\prime}}^{\otimes N}\Vert_1\right)\\
&\leq \frac{1}{2} \log (d-1) \mathbb{E}_{x,x^{\prime}}\Vert \rho_{x}^{\otimes N} - \rho_{x^{\prime}}^{\otimes N} \Vert_1 + h_2\left(\frac{1}{2}\mathbb{E}_{x^{\prime}}\Vert \mathbb{E}_x\rho_{x}^{\otimes N} - \rho_{x^{\prime}}^{\otimes N} \Vert_1\right).\\
&\leq \frac{1}{2} \log (d-1) \mathbb{E}_{x,x^{\prime}}\Vert \rho_{x}^{\otimes N} - \rho_{x^{\prime}}^{\otimes N} \Vert_1 + h_2\left(\min\{ \frac{1}{2},\frac{1}{2}\mathbb{E}_{x^{\prime},x}\Vert\rho_{x}^{\otimes N} - \rho_{x^{\prime}}^{\otimes N} \Vert_1\}\right).
\end{split}
\end{equation}
The second inequality is due to the concavity of $h_2$. Note that
\begin{equation}
\Vert \rho_{x}^{\otimes N} - \rho_{x^{\prime}}^{\otimes N} \Vert_1 \leq \min\{2, \sum_{1\leq i\leq N} \Vert\rho^{\otimes (i-1)}(\rho-\sigma)\sigma^{\otimes(N-i)}\Vert_1\} \leq \min\{2, N\Vert\rho-\sigma\Vert_1\}.
\end{equation}
Thus,
\begin{equation}
I(X:Y^{(N)}) \leq \frac{1}{2} N \log (d-1) \mathbb{E}_{x,x^{\prime}}\Vert \rho_{x} - \rho_{x^{\prime}} \Vert_1 + h_2\left(\min\{\frac{1}{2}, \frac{1}{2}N\mathbb{E}_{x,x^{\prime}}\Vert \rho_{x} - \rho_{x^{\prime}}\Vert_1\}\right).
\end{equation}
One can replace $\log(d-1)$ with the qubit number $n$ without affecting the effectiveness of the bound.

Note that Bob will perform a POVM $\mathfrak{F} = \{F_y\}$ as long as he receives the sending state. Denote the measurement channel as $\mathcal{F}$ such that $\mathcal{F}(\rho) = \sum_y \tr(F_y\rho)\ketbra{y}$. Then the communication between Alice and Bob can be reduced to that Alice chooses states from ensemble $\mathfrak{E}^F = \{p_x, \mathcal{F}(\rho_x)\}$ instead of $\mathfrak{E} = \{p_x, \rho_x\}$, and Bob measures in the computational basis $\{\ketbra{y}\}$ instead of $\mathfrak{F}$. Thus, the upper bound can be further tightened to
\begin{equation}
I(X:Y^{(N)}) \leq \frac{1}{2} N \log (d-1) \mathbb{E}_{x,x^{\prime}}\Vert \mathcal{F}(\rho_x) - \mathcal{F}(\rho_{x^{\prime}}) \Vert_1 + h_2\left(\min\{\frac{1}{2}, \frac{1}{2}N\mathbb{E}_{x,x^{\prime}}\Vert \mathcal{F}(\rho_x) - \mathcal{F}(\rho_{x^{\prime}})\Vert_1\}\right).
\end{equation}
From the above equation, one can obtain that as long as the average pairwise distance of the measured state ensemble $\mathfrak{E}^F = \{p_x, \mathcal{F}(\rho_x)\}$ is exponentially small, even polynomial copies of states cannot deliver any classical message. For further elaboration, we explicitly write down the form of $\mathfrak{E}^F$: $\mathfrak{E}^F = \{p_x, \sum_y \tr(\rho_x F_y)\ketbra{y}\}$.

\subsubsection{Lower bound}
We first provide the single-copy bound where $N=1$. Denote the joint probability distribution of input and output as $P_{XY}$. We have
\begin{equation}
I(X:Y)= D\left(P_{XY} \Vert P_X \otimes P_Y\right) = \mathbb{E}_x D\left(P_{Y|x} \Vert P_Y\right).
\end{equation}
By Pinsker's inequality,
\begin{equation}
D\left(P_{Y|x} \Vert P_Y\right)\geq \frac{1}{2} \Vert P_{Y|x} - P_Y \Vert_1^2 .
\end{equation}

Therefore,
\begin{equation}
\begin{split}
I(X:Y)
&\geq \frac{1}{2} \mathbb{E}_x \Vert P_{Y|x} - P_Y \Vert_1^2 \\
&\geq \frac{1}{2}\left( \mathbb{E}_x \Vert P_{Y|x} - P_Y \Vert_1 \right)^2 \\
&= \frac{1}{2}\left( \mathbb{E}_x\Vert \mathcal{F}(\rho_x) - \mathbb{E}_{x'} \mathcal{F}(\rho_{x'}) \Vert_1\right)^2 \\
&\geq \frac{1}{8}
\left( \mathbb{E}_{x,x'}
\Vert \mathcal{F}(\rho_x) - \mathcal{F}(\rho_{x'}) \Vert_1
\right)^2 .
\end{split}
\end{equation}
Here, the second line utilizes Jensen's inequality, and $\mathcal{F}$ is the measurement channel: $\mathcal{F}(\rho) = \sum_y \tr(F_y\rho)\ketbra{y}$. Thus, one can view $\mathfrak{E}^F = \{p_x, \mathcal{F}(\rho_x)\}$ as a new state ensemble, whose average pairwise distance provides the lower bound for the mutual information. As long as this average pairwise distance is constant, one can use this noisy ensemble to deliver constant bits of information. 

Combining the single-copy lower bound with the upper bound, one can obtain that the the average pairwise distance of $\mathfrak{E}^F = \{p_x, \mathcal{F}(\rho_x)\}$ provides both upper and lower bounds where the gap is at most the qubit number $\log d$. That is, investigate the scaling of $\mathbb{E}_{x,x'}
\Vert \mathcal{F}(\rho_x) - \mathcal{F}(\rho_{x'}) \Vert_1$ suffices for studying that of the mutual information. We'll leave the detailed discussion of the scaling of $\mathbb{E}_{x,x'}
\Vert \mathcal{F}(\rho_x) - \mathcal{F}(\rho_{x'}) \Vert_1$ to the end of this Appendix. Particularly, we find that if the initial state set is a noisy state 2-design $\{p_x, \mathcal{N}(\rho_x)\}$ where $\mathcal{N}$ is the noise channel, then $\mathbb{E}_{x,x'}\Vert \mathcal{F}\circ\mathcal{N}(\rho_x) - \mathcal{F}\circ\mathcal{N}(\rho_{x'}) \Vert_1 = O(\sqrt{\Vert \mathcal{N}^{\otimes 2}(S)\Vert_{\infty}})$, where $\Vert \mathcal{N}^{\otimes 2}(S)\Vert_{\infty}$ is the infinite norm of $\mathcal{N}^{\otimes 2}(S)$. This term is exponentially small in qubit number $n$ if $\mathcal{N}$ is a local purity-shrinking noise. Hence, the mutual information must also be exponentially small and have no noise threshold.

In the following, we derive the $N$-copy lower bound. Note that the random variable of the measurement outcome is $Y^{(N)} = (Y_1, Y_2, \cdots, Y_N)$. We denote the concrete measurement results of $N$ rounds as $\mathbf{y} = (y_1,\ldots,y_N)$. From $\mathbf{y}$, the best candidate Bob can guess for Alice's message is the output of the maximum-likelihood decoder:
\begin{equation}
\bar{x}
= \arg\max_x \Pr(\mathbf{y} | x) .
\end{equation}

Define $A_x$ to be the event that the maximum-likelihood decoder gives a wrong decoding result for input 
$x$:
\begin{equation}
A_x = \{ \mathbf{y} : \exists x' \neq x \text{ s.t. } \Pr(\mathbf{y}|x') \ge \Pr(\mathbf{y}|x) \}.
\end{equation}

The overall error probability is
\begin{equation}
\begin{split}
P_e &= \sum_x P(A_x | x)\cdot P(x) \\
&=\sum_{x}\sum_{\mathbf{y}} \mathbf{1}\{\exists x'\neq x \text{ s.t. } \Pr(\mathbf{y}|x') \ge \Pr(\mathbf{y}|x)\}P(\mathbf{y}|x)P(x).
\end{split}
\end{equation}

By the union bound,
\begin{equation}
\begin{split}
P_e&\leq \sum_x \sum_{\mathbf{y}} \sum_{x' \neq x}
\mathbf{1}\{ P(\mathbf{y} | x') \geq P(\mathbf{y} | x) \}P(\mathbf{y} | x) P(x) \\
&\leq \sum_x \sum_{\mathbf{y}} \sum_{x' \neq x} \sqrt{\frac{P(\mathbf{y}|x')}{P(\mathbf{y}|x)}}
\mathbf{1}\{ P(\mathbf{y} | x') \geq P(\mathbf{y} | x) \}P(\mathbf{y} | x) P(x) \\
&\leq \sum_x \sum_{\mathbf{y}} \sum_{x' \neq x}\sqrt{P(\mathbf{y} | x') P(\mathbf{y} | x)} \, P(x) .
\end{split}
\end{equation}
The second line is because $\sqrt{\frac{P(\mathbf{y}|x')}{P(\mathbf{y}|x)}}\ge 1$ whenever the indicator is nonzero. Thus,
\begin{equation}
\begin{split}
P_e&\leq \sum_{x, x' \neq x} P(x)\sum_{\mathbf{y}} \sqrt{P(\mathbf{y} | x') P(\mathbf{y} | x)} \\
&= \sum_{x, x' \neq x} P(x)\left( \sum_y \sqrt{P(y | x') P(y | x)} \right)^N \\
&= \sum_{x, x' \neq x} P(x)B^N\left(P(\cdot | x'), P(\cdot | x)\right) .
\end{split}
\end{equation}

The classical fidelity is
\begin{equation}
B\left(P(\cdot | x), P(\cdot | x')\right)
= \sum_y \sqrt{P(y | x) P(y | x')} .
\end{equation}
By the Fuchs--van de Graaf inequality,
\begin{equation}
\begin{split}
B\left(P(\cdot | x), P(\cdot | x')\right)
&\leq \sqrt{1 - \frac{1}{4}
\Vert P(\cdot | x) - P(\cdot | x') \Vert_1^2} \\
&= \sqrt{1 - \frac{1}{4}
\Vert \mathcal{F}(\rho_x) - \mathcal{F}(\rho_{x'}) \Vert_1^2} .
\end{split}
\end{equation}

Define
\begin{equation}
\delta_{\min}^{\mathcal{F}}= \min_{x \neq x'} \Vert \mathcal{F}(\rho_x) - \mathcal{F}(\rho_{x'}) \Vert_1 / 2.
\end{equation}
Suppose $M=|X|$ is the size of the message alphabet. Then
\begin{equation}
P_e\leq (M-1)\left( 1 - (\delta_{\min}^{\mathcal{F}})^2 \right)^{\frac{N}{2}} .
\end{equation}
By Fano's inequality,
\begin{equation}
\begin{split}
I(X : Y^{(N)})&\geq H(X) - h_2(P_e) - P_e \log M \\
&\geq H(X) - 1 - (M-1) \left( 1 - (\delta_{\min}^{\mathcal{F}})^2 \right)^{\frac{N}{2}}\log M .
\end{split}
\end{equation}
Thus, as long as there exists a POVM such that the minimum pairwise distance is at least polynomial small, one can ensure that polynomial copies can enable one to obtain a mutual information comparable to $H(X)$. The measurement channel $\mathcal{F}$ will be optimized to provide a tight lower bound.

Note that for the upper bound and the single-copy lower bound cases, we use the average pairwise distance to provide
bounds. The $N$-copy lower bound case is different and uses the minimum pairwise distance. We believe this difference comes from a technical reason and expect that one can provide a lower bound for the $N$-copy case with the average pairwise distance. We leave this for future work.

\subsubsection{Noisy state ensemble as POVM}
Here, we discuss the duality case of the above results and use the pairwise distance of POVM elements to provide bounds for mutual information. In this case, Alice still uses state ensemble $\{p_x, \rho_x\}$ to encode message $X$ where Bob performs a fixed noisy POVM $\mathfrak{F} = \{F_y\}$. One can view the POVM element as $F_y = \mathcal{N}^{\dagger}(F_y^0)$ as the original POVM element passing through the (adjoint) noise channel $\mathcal{N}^{\dagger}$. Notice that $\mathcal{N}^{\dagger}$ is always unital. Note that the maximal achievable mutual information considering all input states and a fixed POVM is the information power of this POVM~\cite{DallArno2011power}.

Denote $\rho = \sum_x p_x\rho_x$, and define new state ensemble and POVM:
\begin{equation}
\mathfrak{E}_{\mathfrak{F}} = \{\tr \rho F_y, \frac{\sqrt{\rho} F_y \sqrt{\rho}}{\tr \rho F_y}\}, \mathfrak{F}_{\mathfrak{E}} = \{p_x \rho^{-1/2}\rho_x\rho^{-1/2}\}.
\end{equation}
The new state ensemble and POVM generate the same probability distribution as the original one:
\begin{equation}
P_{XY}(x, y) = \tr \rho F_y \tr \frac{\sqrt{\rho} F_y \sqrt{\rho}}{\tr \rho F_y} p_x \rho^{-1/2}\rho_x\rho^{-1/2} = p_x\tr \rho_xF_y.
\end{equation}
Thus, we can use the average pairwise distance of the new state ensemble to provide bounds on mutual information:
\begin{equation}
\mathbb{E}_{yy'}\Vert \frac{\sqrt{\rho} F_y \sqrt{\rho}}{\tr \rho F_y} - \frac{\sqrt{\rho} F_{y'} \sqrt{\rho}}{\tr \rho F_{y'}} \Vert_1 / 2.
\end{equation}
Same as the derivation before, we can add the new measurement channel $\mathcal{M}(\sigma) = \sum_x \tr(\sigma p_x \rho^{-1/2}\rho_x\rho^{-1/2})\ketbra{x}$ to the state ensemble and evaluate a tighter bound. After adding this measurement channel, the state ensemble becomes 
\begin{equation}
\mathfrak{E}_{\mathfrak{F}}^{\mathcal{M}} = \{\tr \rho F_y, \mathcal{M}(\frac{\sqrt{\rho} F_y \sqrt{\rho}}{\tr \rho F_y})\} = \{\tr \rho F_y, \sum_x  p_x\frac{\tr \rho_x F_y }{\tr \rho F_y} \ketbra{x}\}.
\end{equation}
We can use the average pairwise distance of $\mathfrak{E}_{\mathfrak{F}}^{\mathcal{M}}$ to provide the upper bound of the mutual information. A simple case is where the ensemble $\mathfrak{E} = \{p_x, \rho_x\}$ forms a state 1-design, or $\sum_x p_x\rho_x = \id/d$. In this case, we have that
\begin{equation}
\mathfrak{E}_{\mathfrak{F}}^{\mathcal{M}} = \{\frac{\tr F_y}{d}, \mathcal{M}(\frac{F_y}{\tr F_y})\},
\end{equation}
where
\begin{equation}
\mathcal{M}(\sigma) = \sum_x \tr(\sigma d p_x \rho_x)\ketbra{x}.
\end{equation}
If the POVM $\mathfrak{F}$ is a 2-design, or $\{\frac{\tr F_y}{d}, \frac{F_y}{\tr F_y}\}$ is a state 2-design, $\sum_y \frac{\tr F_y}{d} (\frac{F_y}{\tr F_y})^{\otimes 2} = \frac{\id+S}{d(d+1)}$, the average pairwise distance of $\mathfrak{E}_{\mathfrak{F}}^{\mathcal{M}}$ will have a concise analytic expression, which will be discussed in later subsections.

\subsection{Accessible information of noisy 2-design}\label{appendssc:mulinf2design}
Below, we use the average pairwise distance of a noisy 2-design ensemble to provide upper bounds for its accessible information or information power. As mentioned before, when using the noisy 2-design to encode classical information, the upper bound of the accessible information can be given by
\begin{equation}
I(X:Y^{(N)}) \leq \frac{1}{2} N \log (d-1) \mathbb{E}_{x,x^{\prime}}\Vert \mathcal{F}(\rho_x) - \mathcal{F}(\rho_{x^{\prime}}) \Vert_1 + h_2\left(\min\{\frac{1}{2}, \frac{1}{2}N\mathbb{E}_{x,x^{\prime}}\Vert \mathcal{F}(\rho_x) - \mathcal{F}(\rho_{x^{\prime}})\Vert_1\}\right).
\end{equation}
Note that $\rho_x = \mathcal{N}(\psi_x)$ with $\psi_x$ forming a state 2-design. By Theorem~\ref{thm:2design}, we have that
\begin{equation}
\mathbb{E}_{x,x^{\prime}}\Vert \mathcal{F}(\rho_x) - \mathcal{F}(\rho_{x^{\prime}}) \Vert_1/2 \leq 2^{-\frac{1}{2}H_2(S|S')_{\mathcal{F}\circ\mathcal{N}}}.
\end{equation}
The following is to evaluate the value of 
$H_2(S|S')_{\mathcal{F}_{S'}\circ\mathcal{N}_{S'}(\ketbra{\Phi}_{SS'})}$. Recall that $\mathcal{F}(\rho) = \sum_y \tr(F_y \rho)\ketbra{y}$. Thus, $\mathcal{F}\circ\mathcal{N}(\rho)  = \sum_y \tr(\mathcal{N}^{\dagger}(F_y) \rho)\ketbra{y}$ where $\mathcal{N}^{\dagger}$ is the adjoint channel of $\mathcal{N}$. Thus,
\begin{equation}
\begin{split}
\mathcal{F}_{S'}\circ\mathcal{N}_{S'}(\ketbra{\Phi}_{SS'}) &= \frac{1}{d}\sum_y (\mathcal{N}^{\dagger}(F_y))^T_{S}\otimes \ketbra{y}_{S'}\\
&= \sum_y p_y \mathcal{N}^{\dagger}(\rho_y)_{S} \otimes \ketbra{y}_{S'},
\end{split}
\end{equation}
where we set $p_y = \tr F_y / d$ and $\rho_y = F_y^T/\tr F_y$. Note that $H_2(S|S')_{\rho_{SS"}} = -\log \tr[ ( \rho_{S'}^{-1/4}\rho_{SS'}\rho_{S'}^{-1/4} )^2 ]$. Hence,
\begin{equation}
\begin{split}
H_2(S|S')_{\mathcal{F}_{S'}\circ\mathcal{N}_{S'}(\ketbra{\Phi}_{SS'})} &= -\log \tr\sum_y \frac{p_y (\mathcal{N}^{\dagger}(\rho_y))^2}{\tr \mathcal{N}^{\dagger}(\rho_y)}\otimes \ketbra{y}\\
&= -\log \sum_y \frac{p_y \tr\Bigl( (\mathcal{N}^{\dagger}(\rho_y))^2\Bigr)}{\tr \mathcal{N}^{\dagger}(\rho_y)}\\
&= -\log \sum_y p_y \frac{\tr\rho_y^{\otimes 2}\mathcal{N}^{\otimes 2}(S)}{\tr \rho_y \mathcal{N}(\id)}.
\end{split}
\end{equation}
Thus,
\begin{equation}
\mathbb{E}_{x,x^{\prime}}\Vert \mathcal{F}(\rho_x) - \mathcal{F}(\rho_{x^{\prime}}) \Vert_1 / 2 \leq \sqrt{\sum_y p_y \frac{\tr\rho_y^{\otimes 2}\mathcal{N}^{\otimes 2}(S)}{\tr \rho_y \mathcal{N}(\id)}}.
\end{equation}
Thus, for a unital channel $\mathcal{N}$, we have that
\begin{equation}
\mathbb{E}_{x,x^{\prime}}\Vert \mathcal{F}(\rho_x) - \mathcal{F}(\rho_{x^{\prime}}) \Vert_1 / 2 \leq \sqrt{\Vert\mathcal{N}^{\otimes 2}(S) \Vert_{\infty}}
\end{equation}
\begin{equation}
I(X:Y^{(N)}) \leq N \log (d-1) \sqrt{\Vert\mathcal{N}^{\otimes 2}(S) \Vert_{\infty}} + h_2\left(\min\{\frac{1}{2}, N\sqrt{\Vert\mathcal{N}^{\otimes 2}(S) \Vert_{\infty}}\}\right).
\end{equation}
For local Pauli noise strictly shrinking the purity of any quantum state, like local depolarizing noise, we have that $\Vert\mathcal{N}^{\otimes 2}(S) \Vert_{\infty} = \exp(-\Theta(n))$. In particular, if $\mathcal{N} = \mathcal{N}_p^{\otimes n}$ with $\mathcal{N}_p(\rho) = (1-p)\rho+p\frac{\id}{2}$, $\Vert\mathcal{N}^{\otimes 2}(S) \Vert_{\infty} = (\frac{1+(1-p)^2}{2})^n$. In this case, the mutual information decays exponentially with respect to the qubit number. By continuity, if the noise is a very noisy local depolarizing channel mixed with a little amplitude-damping noise, the term also decays exponentially fast. As discussed before, by state-measurement duality, if the input state forms a state 1-design and the measurement is a noisy state 2-design, we have exactly the same result. One can repeat the derivation above and obtain that the mutual information is bounded by replacing $\Vert\mathcal{N}^{\otimes 2}(S) \Vert_{\infty}$ with $\Vert\mathcal{N}^{\dagger\otimes 2}(S) \Vert_{\infty}$. For generic input state ensembles, we use another way instead of the average pairwise distance to provide upper bounds, which is shown below.

\subsection{Upper bound of mutual information for noisy two-design POVM}\label{appendssc:mulinf2designPOVM}
In this part, we present another way to upper bound the mutual information given the input state or measurement forming a noisy state 2-design, while the other side can be chosen freely.

Recall that the probability of Bob getting outcome $y$ given Alice sent $x$ is $P(y|x) = \tr(F_y \mathcal{N}(\rho_x))$. The mutual information from single-copy measurement is the expected KL divergence from the marginal:
\begin{equation}
I(X:Y) = \sum_x p_x D_{KL}(P_{Y|x} || P_Y).
\end{equation}
Because KL divergence is relative, we can introduce a reference distribution. Let $P_U(y)$ be the unconditional distribution Bob would get if he measured the maximally mixed state $I/d$:
\begin{equation}
P_U(y) = \tr\left(F_y \frac{I}{d}\right) \equiv p_y.
\end{equation}
A standard information-theoretic inequality states that for any reference distribution $Q$, $D_{KL}(P_{Y|x} || P_Y) \le D_{KL}(P_{Y|x} || Q)$. Furthermore, the KL divergence is strictly upper-bounded by the $\chi^2$-divergence: $I(X:Y) \le \sum_x p_x D_{KL}(P_{Y|x} || P_U) \le \sum_x p_x \chi^2(P_{Y|x} || P_U)$. Thus, it suffices to upper bound $\sum_x p_x \chi^2(P_{Y|x} || P_U)$.

Note that
\begin{equation}\label{eq:chi2bound}
\begin{split}
\sum_x p_x \chi^2(P_{Y|x} || P_U) &= \sum_x p_x \sum_y \frac{(P(y|x) - P_U(y))^2}{P_U(y)}\\
&= \sum_x p_x \sum_y \frac{(\tr(F_y \mathcal{N}(\rho_x)) - p_y)^2}{p_y}\\
&= \sum_{x,y} \frac{p_x(\tr(F_y \mathcal{N}(\rho_x)))^2}{p_y} - 1.
\end{split}
\end{equation}
If the POVM $\{p_y = \frac{\tr F_y}{d}, \frac{F_y}{\tr F_y}\}$ forms a state 2-design, then the above equation reduces to
\begin{equation}
\sum_x p_x \chi^2(P_{Y|x} || P_U) = \frac{d\sum_x p_x\tr (\mathcal{N}(\rho_x))^2 - 1}{d+1}\leq \sum_x p_x\tr (\mathcal{N}(\rho_x))^2 \leq \Vert \mathcal{N}^{\dagger\otimes 2}(S) \Vert_{\infty},
\end{equation}
no matter what the input state ensemble is. $\mathcal{N}^{\dagger}$ is the adjoint map of $\mathcal{N}$. If Alice sends $N$ copies to Bob, then the total mutual information will be bounded by $N \Vert \mathcal{N}^{\dagger\otimes 2}(S) \Vert_{\infty}$. For local depolarizing noise, $\mathcal{N} = \mathcal{N}^{\dagger}$, and $\Vert \mathcal{N}^{\dagger\otimes 2}(S) \Vert_{\infty} = \exp(-\Theta(n))$ will decay exponentially in the qubit number.

Note that the above derivation can be repeated for the noisy 2-design input. In Eq.~\eqref{eq:chi2bound}, if $\{p_x, \rho_x\}$ forms a state 2-design, we can obtain that the upper bound for the mutual information is $\frac{d}{d+1}\Vert \mathcal{N}^{\otimes 2}(\id+S) \Vert_{\infty}-1$. For the unital noise, $\frac{d}{d+1}\Vert \mathcal{N}^{\otimes 2}(\id+S) \Vert_{\infty}-1\leq \frac{d}{d+1}(1+\Vert \mathcal{N}^{\otimes 2}(S)\Vert_{\infty})-1$, the mutual information can be bounded by $\Vert \mathcal{N}^{\otimes 2}(S) \Vert_{\infty}$.

\section{Tailored upper bound on mutual information under noisy 2-design POVM for universal learning protocols}\label{appendsc:universal}
In this part, we derive a tailored upper bound on mutual information for universal learning protocols like classical shadow tomography. For simplicity, we consider that the learning protocol has already been embedded in a communication protocol for measuring states and predicting input messages. The tailored upper bound can be tighter than the one derived in the previous section. Different from previous discussions, where we upper bound the mutual information between input $X$ and measurement result $Y$, here we assume Bob stores $Y$ and output an estimate $\Bar{X}$ of $X$ via a universal learning protocol. We would like to upper bound the mutual information of $I(X:\Bar{X})$. In this procedure, we have some requirements for the estimation process or the learning process of $\Bar{X}$ from $Y$. We will restate the problem and introduce the details below.

Consider a quantum communication protocol between Alice and Bob involving an $n$-qubit quantum system:
\begin{enumerate}
\item Alice selects a classical random variable $X$, which equals $x$ with probability distribution $p_X(x)$. She prepares an $n$-qubit state $\rho_X$ to encode $X$, and she sends $N$ copies of $\rho_X$ to Bob. Here, we require that any state $\rho_x$ can be transformed into the same state $\rho$ by a unitary operation, $\rho_x = U_x\rho U_x^{\dagger}$, where $U_x$ can be chosen freely.
\item For each copy of $\rho_X$, Bob individually performs the POVM $\mathfrak{F} = \{F_y\}$ on the received state to obtain measurement results. The elements satisfy $F_y \geq 0$, $F_y^{\dagger} = F_y$, and $\sum_y F_y = \id$. The group of all measurement outcomes is denoted as $Y = Y^{(N)} = (Y_1, Y_2, \cdots, Y_N)$ where $Y_i$ is the $i$-th measurement outcome.
\item Bob processes $Y$ to estimate several functions of $\rho_X$, $\{f_i(\rho_X)\}$. The function $f_i$ should satisfy that one can get the result of $f_i(\rho_X)$ with the measurement result of $U\rho_X U^{\dagger}$ and the knowledge of $U$ for any unitary operation $U$.
\item Bob uses the estimation of $\{f_i(\rho_X)\}$ to predict the estimate $\Bar{X} = g(f_i(\rho_X))$.
\end{enumerate}
Here, the learning protocol has already been embedded in steps 2 to 4. For classical shadow tomography, step 2 is performed with a 2-design ensemble; step 3 is estimating linear observables of $\rho_X$; step 4 is using linear observable expectation results to estimate $\Bar{X}$. The condition $\rho_x = U_x\rho U_x^{\dagger}$ in the first point and the requirement of the function $f_i$ in the third point are needed because we would like to introduce the random unitary operation in the derivation. The requirement of $f_i$ given by the third point is the reason why we call the learning protocol to be ``universal": Bob should recover $f_i(\rho_X)$ no matter what $U$ is. The first condition is automatically satisfied if Alice encodes the classical messages using pure states. The second condition will be explained later. We first introduce a new communication protocol, in which the achievable mutual information is always higher than that in the above communication protocol. The new protocol introduces a hypothetical adversary, Loki, to intercept the sending state and apply a random unitary operation.
\begin{enumerate}
\item Alice selects a classical random variable $X$, which equals $x$ with probability distribution $p_X(x)$. She prepares an $n$-qubit state $\rho_X$ to encode $X$, and she sends $N$ copies of $\rho_X$ to Bob. Here, we require that any state $\rho_x$ can be transformed into the same state $\rho$ by a unitary operation, $\rho_x = U_x\rho U_x^{\dagger}$.
\item Loki intercepts the state $\rho_X$ and applies a random unitary $U$ sampled from the Haar measure before Bob measures it.
\item For each copy of $U\rho_X U^{\dagger}$, Bob individually performs the POVM $\mathfrak{F} = \{F_y\}$ on the received state to obtain measurement results. The elements satisfy $F_y \geq 0$, $F_y^{\dagger} = F_y$, and $\sum_y F_y = \id$. The group of all measurement outcomes is denoted as $Y^U$.
\item After Bob obtains the measurement outcome, Loki reveals the unitary $U$ to Bob.
\item Bob processes $Y^U$ and $U$ to estimate several functions of $\rho_X$, $\{f_i(\rho_X)\}$.
\item Bob uses the estimation of $\{f_i(\rho_X)\}$ to predict the estimate $\Bar{X} = g(f_i(\rho_X))$.
\end{enumerate}
In this modified communication protocol, we have that
\begin{itemize}
\item The state received by Bob becomes $\rho_{X,U} = U \rho_X U^\dagger$.
\item The probability of outcome $y$, conditioned on the message $X$ and the unitary $U$, is given by $p(y|X,U) = \tr(F_y U \rho_X U^\dagger)$.
\end{itemize}
Thanks to the second condition, even if Loki applies the random unitary operation, Bob can still get the result of $f_i(\rho_X)$ with the measurement result of $U\rho_X U^{\dagger}$ and the knowledge of $U$. Thus, the second communication protocol can always be reduced to the first communication protocol and has higher achievable mutual information.

More specifically, for the second condition, we assume that there exists an estimation process that provides $\hat{f}_i(\rho_X)$ with the outcome of the measurement of step 2 as input, such that $\abs{\hat{f_i}(\rho_X)-f_i(\rho_X)}<\epsilon$ with a probability greater than $1-\delta$ for any $X$, and we call such an estimation a $(\epsilon,\delta)$ estimation. The process is required to be reversible to adversarial unitary evolution imposed by Loki. That is, if $\rho_X$ is changed into $U\rho_XU^{\dagger}$, then we can also get an $(\epsilon,\delta)$-estimation of $f_i(\rho_X)$ with the knowledge of $U$. In the last step, we get $\Bar{X}$ via a function $g$ of $f_i$'s, $\Bar{X}=g(f_i(\rho_X))$. We require $g$ satisfying $g(f_i(\rho_X))=g(\hat{f}_i(\rho_X))$ for all $X$ and any $\epsilon$-estimation $\hat{f}_i$. In this case, as long as Bob can estimate $f_i$ to accuracy $\epsilon$, he can obtain the correct prediction of $X$.

The learning protocol of classical shadow tomography~\cite{huang2020shadow} satisfies the assumptions above. In this case, the function $f_i$ is set as a linear function of the input state $\tr \rho O$. One ensures that for the input state ensemble $\{\rho_i\}$, we have a set of linear observables $\{O_j\}$ such that for any $i$ and $j\neq i$,
\begin{equation}
\abs{\tr(\rho_iO_i) - \tr(\rho_i O_j)} > 3\epsilon.
\end{equation}
Thus, as long as one estimates all the expectation values of linear observables $\{O_j\}$ of the input state within accuracy $\epsilon$, one can determine the index $i$ and the input state of Alice. This corresponds to the condition $g(f_i(\rho_X))=g(\hat{f}_i(\rho_X))$. Meanwhile, the classical shadow is a universal learning protocol, which estimates any linear observable of the measured state within accuracy $\epsilon$. Hence, with measurement result of $U\rho_X U^{\dagger}$ and knowledge of $U$, Bob can estimate the observable $UO_i U^{\dagger}$, or $\tr U\rho_X U^{\dagger} UO_i U^{\dagger}$, to get the estimate of $\tr \rho_X O_i = f_i(\rho_X)$. Thus, the second condition is satisfied by the classical shadow. Note that our analysis applied to all universal learning protocols satisfying the second condition, regardless of whether they estimate linear observables or not.

In the second communication protocol, Bob gets an $(\epsilon,\delta)$-estimation $\hat{f}'_i(\rho_X)$ with $Y^U$ and $U$, then gets $\Bar{X}'$ in the last step. Data processing inequality implies $I(X:Y^U,U)\geq I(X:\hat{f}'_i(\rho_X))\geq I(X:\Bar{X}')$. Let $\Bar{X}$ be the prediction of $X$ without Loki. By assumption, we have $p_e=\Pr(\Bar{X}'\neq \Bar{X})<\Pr(\Bar{X'}\neq g(f_i(\rho_X)))+\Pr(\Bar{X}\neq g(f_i(\rho_X)))<2\delta$. Fano's inequality implies $H(\Bar{X}|\Bar{X}')\leq H(p_e)+p_e\log M$. By the chain rule
\begin{align}
&I(X:\Bar{X},\Bar{X}')=I(X:\Bar{X})+I(X:\Bar{X}'|\Bar{X})\geq I(X:\Bar{X})\\
&I(X:\Bar{X},\Bar{X}')=I(X:\Bar{X}')+I(X:\Bar{X}|\Bar{X}')\leq I(X:\Bar{X}')+H(\Bar{X}|\Bar{X}')\\
\Rightarrow &I(X:\Bar{X})\leq I(X:\Bar{X}')+H(\Bar{X}|\Bar{X}')\leq I(X:Y^U,U)+H(p_e)+p_e\log M\leq I(X:Y^U|U)+H(2\delta)+2\delta\log M
\end{align}


In the concrete protocol, the success probability $\delta$ always decays exponentially in the qubit number $n$, and $\log M$ is normally a polynomial scaling of qubit number $n$. For instance, the stabilizer state satisfies $\log M = O(n^2)$. Hence, in the large $n$ limit, we can simplify the above result to
\begin{equation}
I(X:\Bar{X})\leq I(X:Y^U|U).
\end{equation}
Thus, to derive an upper bound on the mutual information $I(X:\Bar{X})$ when Alice's state or Bob's POVM is noisy, the key point is to upper bound $I(X:Y^U|U)$. Note that $I(X:Z|U) = I(X,U:Z)-I(U:Z)\leq I(X,U:Z)$. Thus,
\begin{equation}
I(X:\Bar{X})\leq NI(X,U:\mathfrak{F}\text{ on }U\rho_XU^{\dagger})
\end{equation}
The problem reduces to bounding the mutual information $I(X,U:\mathfrak{F}\text{ on }U\rho_XU^{\dagger})$. Utilizing the inequality $\log(x) \leq \log(y) + \frac{x-y}{y}$, we arrive at:
\begin{equation}
\begin{split}
I(X,U:\mathfrak{F} \text{ on } U\rho_{X}U^{\dagger}) &= H(\mathfrak{F} \text{ on } U\rho_{X}U^{\dagger})-H(\mathfrak{F} \text{ on } U\rho_{X}U^{\dagger}|X,U)\\
&=H\left(\mathbb{E}_{X,U}[\tr(U\rho_{X}U^{\dagger}\mathfrak{F})]\right)-\mathbb{E}_{X,U}[H(\tr(U\rho_{X}U^{\dagger}\mathfrak{F}))] \\
&= H\left(\mathbb{E}_{U}[\tr(U\rho U^{\dagger}\mathfrak{F})]\right)-\mathbb{E}_{U}[H(\tr(U\rho U^{\dagger}\mathfrak{F}))]\\
&= \sum_{y}[-\mathbb{E}_{U}p(y|U)\log(\mathbb{E}_{U}p(y|U))]+\mathbb{E}_{U}[p(y|U)\log p(y|U)] \\
&\le \sum_{y}[-\mathbb{E}_{U}p(y|U)\log(\mathbb{E}_{U}p(y|U))]+\mathbb{E}_{U}\left[p(y|U)\log(\mathbb{E}_{U}p(y|U))+p(y|U)\frac{p(y|U)-\mathbb{E}_{U}p(y|U)}{\mathbb{E}_{U}p(y|U)}\right] \\
&= \left(\sum_{y}\frac{\mathbb{E}_{U}[p(y|U)^{2}]}{\mathbb{E}_{U}[p(y|U)]}\right)-1
\end{split}
\end{equation}
Here, $p(y|U) = \mathbb{E}_Xp(y|X,U) = \tr(F_y U \rho U^\dagger)$. The third line holds because there exists a unitary $U_X$ such that $\rho_X = U_X\rho U_X^{\dagger}$.

Thus, no matter how Alice encodes her classical message $X$, the mutual information obtained by Bob applying POVM $\mathfrak{F}$ is bounded by
\begin{equation}
I(X,U:\mathfrak{F} \text{ on } U\rho_{X}U^{\dagger}) \leq \sum_{y}\frac{\mathbb{E}_{U}[p(y|U)^{2}]}{\mathbb{E}_{U}[p(y|U)]}-1.
\end{equation}
Thus,
\begin{equation}
I(X:\Bar{X}) \leq N\left( \sum_{y}\frac{\mathbb{E}_{U}[p(y|U)^{2}]}{\mathbb{E}_{U}[p(y|U)]}-1 \right).
\end{equation}
This upper bound can give us a tighter result compared to the previous upper bound on $I(X:Y)$. Below, we consider the POVM to be a noisy state 2-design ensemble and derive the bound.

In particular, we set $\mathfrak{F} = \{\mathcal{N}(F_y)\}$, where $\sum_{y}\frac{\tr F_y}{d} \frac{F_y}{\tr F_y} = \frac{\id}{d}$, $\sum_{y}\frac{\tr F_y}{d} (\frac{F_y}{\tr F_y})^{\otimes 2} = \frac{\id+S}{d(d+1)}$, and $\mathcal{N}$ is a unital quantum channel. In this case, $\mathfrak{F} = \{\mathcal{N}(F_y)\}$ is still a valid POVM. As a result,
\begin{equation}
\mathbb{E}_{U}[p(y|U)] = \tr(\mathcal{N}(F_y))/d = \tr(F_y)/d.
\end{equation}
\begin{equation}
\begin{split}
\mathbb{E}_{U}[p(y|U)^2] &= \mathbb{E}_U[\tr( \mathcal{N}^{\otimes 2}(F_y^{\otimes 2})U^{\otimes 2} \rho^{\otimes 2} U^{\dagger\otimes 2})]\\
&= \tr( \mathcal{N}^{\otimes 2}(F_y^{\otimes 2}) \left[\frac{1-d^{-1}\tr\rho^2}{d^2-1}\id+\frac{\tr\rho^2-d^{-1}}{d^2-1}S\right] ),
\end{split}
\end{equation}
\begin{equation}
\begin{split}
\sum_{y}\frac{\mathbb{E}_{U}[p(y|U)^{2}]}{\mathbb{E}_{U}[p(y|U)]} &= \sum_y d\frac{\tr( \mathcal{N}^{\otimes 2}(F_y^{\otimes 2}) \left[\frac{1-d^{-1}\tr\rho^2}{d^2-1}\id+\frac{\tr\rho^2-d^{-1}}{d^2-1}S\right] )}{\tr(F_y)} \\
&= d^2 \tr( \mathcal{N}^{\otimes 2}(\frac{\id+S}{d(d+1)}) \left[\frac{1-d^{-1}\tr\rho^2}{d^2-1}\id+\frac{\tr\rho^2-d^{-1}}{d^2-1}S\right] )\\
&\leq \tr( \mathcal{N}^{\otimes 2}(\frac{\id+S}{d(d+1)}) (\id+S) )\\
&= 1+\frac{1}{d+1}+\frac{\tr(\mathcal{N}^{\otimes 2}(S)S)}{d(d+1)}.
\end{split}
\end{equation}
Note that $\tr(\mathcal{N}^{\otimes 2}(S)S) = d^2\tr(\tau_{\mathcal{N}}^2)$ where $\tau_{\mathcal{N}}$ is the Choi state of $\mathcal{N}$. Hence,
\begin{equation}
I(X:\Bar{X}) \leq N\frac{d\tr(\tau_{\mathcal{N}}^2)+1}{d+1}.
\end{equation}
For local Pauli noise like depolarizing and dephasing, there exists some constant $p$, $\tr(\tau_{\mathcal{N}}^2) = p^n$. In this case, $I(X:\Bar{X}) \sim \exp(-\Theta(n))$. Note that this bound is determined by the purity of $\tau_{\mathcal{N}}$ and is always smaller than the previous bound $\Vert \mathcal{N}^{\otimes 2}(S)\Vert_{\infty}$. Thus, for the universal learning protocol, we obtain a strictly tighter bound for its error robustness.

\end{document}